\documentclass[journal]{IEEEtran}
\ifCLASSINFOpdf

\else
 \fi
 \usepackage{amsmath}
\usepackage{amsfonts}
\usepackage{amsthm,mathrsfs}
\usepackage{dsfont}
\usepackage{amssymb}
\usepackage{amstext}
\usepackage{color, epsf, epsfig, graphicx, psfrag, subfigure}
\usepackage{color,cite,times,latexsym,ifthen}

\newcommand{\SNRP}{{\mathsf{SNR}}}
\newcommand{\INRP}{{\mathsf{INR}}}
\newcommand{\SNR}{\mathsf{SNR}}
\newcommand{\INR}{\mathsf{INR}}
\newcommand{\DOF}{\mathbf{GDoF}}

\newcommand{\E}{\mathbb{E}}

\newcommand{\bx}{\mathbf{x}}
\newcommand{\by}{\mathbf{y}}
\newcommand{\bz}{\mathbf{z}}
\newcommand{\bs}{\mathbf{s}}
\newcommand{\cM}{\mathcal{M}}

\newcommand{\cgf}[2]{\frac{1}{2}\log\left( \frac{#1}{#2}\right)}
\newcommand{\cg}[1]{\frac{1}{2}\log\left( #1 +1\right)}
\newcommand{\cgfp}[2]{\frac{1}{2}\log^+\left( \frac{#1}{#2}\right)}
\newtheorem{Theo}{Theorem}
\newtheorem{Rem}{Remark}

\newtheorem{Cor}{Corollary}

\newtheorem{claim}{Claim}

\begin{document}

\title{On the Symmetric Feedback Capacity of the $K$-user Cyclic Z-Interference Channel}
\author{Ravi Tandon, \IEEEmembership{Member, IEEE},
Soheil~Mohajer, 
and 
H. Vincent Poor, \IEEEmembership{Fellow,  IEEE}
  \thanks{Manuscript received September 3, 2011; revised July 27, 2012;  accepted December 15, 2012. }
  \thanks{{\footnotesize

      Ravi Tandon  was with the Department of Electrical Engineering, 
      Princeton University, Princeton, NJ, USA. He is now with the
     Department of Electrical and Computer Engineering, Virginia Tech, 
     Blacksburg, VA, USA. (E-mail: tandonr@vt.edu).
            
      Soheil Mohajer  was with the Department of Electrical Engineering, 
      Princeton University, Princeton, NJ, USA. He is now with  the
      Department of Electrical Engineering and Computer Sciences,
      University of California at Berkeley, Berkeley, CA, USA. 
      (E-mail: mohajer@eecs.berkeley.edu).
           
H. Vincent Poor is with the Department of Electrical Engineering, Princeton
     University, Princeton, NJ, USA. (E-mail: poor@princeton.edu).
     
     The work was supported in part by the Air Force Office of Scientific 
     Research under MURI Grant FA-$9550$-$09$-$1$-$0643$ and in 
     part by the DTRA under Grant HDTRA-$07$-$1$-$0037$. The work 
     of Soheil Mohajer is partially supported by The Swiss National Science 
     Foundation under Grant PBELP2-$133369$.  This paper was presented in part at 49th 
     Annual Allerton Conference on Communications, Control and Computing, Monticello, IL, September 2011. }}  
     \thanks{Copyright (c) 2012 IEEE. Personal use of this material is permitted.  However, permission to use this material for any other purposes must be obtained from the IEEE by sending a request to pubs-permissions@ieee.org.}}

\markboth{IEEE Transactions on Information Theory,~Vol.~x, No.x~, Month~20xx}%
{Shell \MakeLowercase{\textit{et al.}}: Bare Demo of IEEEtran.cls for Journals}

\maketitle

\begin{abstract}
The $K$-user cyclic Z-interference channel models a situation in which the $k$th transmitter causes interference only to the $(k-1)$th receiver in a cyclic manner, e.g., the first transmitter causes interference only to the $K$th receiver.
The impact of noiseless feedback on the capacity of this channel is studied by focusing on the Gaussian cyclic Z-interference channel.
To this end, the symmetric feedback capacity of the linear shift deterministic cyclic Z-interference channel (LD-CZIC) is completely characterized for all interference regimes. Using insights from the linear deterministic channel model, the symmetric feedback capacity of the Gaussian cyclic Z-interference channel is characterized up to within a constant number of bits. As a byproduct of the constant gap result, the symmetric generalized degrees of freedom with feedback for the Gaussian cyclic Z-interference channel are also characterized. These results highlight that the symmetric feedback capacities for both linear and Gaussian channel models are in general functions of $K$, the number of users. Furthermore, the capacity gain obtained due to feedback decreases as $K$ increases. 
\end{abstract}

\section{Introduction}
Managing the effects of interference is a key issue in currently deployed wireless networks. Among several ways to mitigate or
perhaps constructively using interference is to make use of cooperation amongst interfering users. In this paper, we focus on one
such important issue by studying the impact of noiseless receiver-to-transmitter feedback on the capacity of the $K$-user cyclic Z-interference
channel (CZIC). In this model, $K$ transmitters intend to transmit independent messages to $K$ respective receivers and the $k$th
transmitter causes interference to the $(k-1)$th receiver in a cyclic manner. The motivation for studying the cyclic Z-interference channel comes from the modified Wyner model \cite{SomekhWyner2007}, which describes the soft handoff scenario of a cellular network. In the original Wyner model \cite{WynerModel1994},
each receiver can suffer interference from its adjacent transmitters. In the modified Wyner model, one can assume that
the terminals are situated along a circular array (see Figure \ref{figurewynermodel}). If in addition, we assume that the
mobile communicates with the intended base-station on its left (or right), while suffering interference due to the BS
to its right (or left), then the resulting channel model is the $K$-user CZIC, which is considered in this paper. The $K$-user Gaussian CZIC (G-CZIC) \emph{without feedback} was recently investigated in \cite{WeiYuCIC}, where it was
shown that the generalized degrees-of-freedom of the symmetric $K$-user G-CZIC are the same as for the $2$-user Gaussian
interference channel. By an interesting generalization of the results of Etkin, Tse and Wang \cite{ETW2008}, the approximate
symmetric capacity was characterized for the weak interference regime and the exact capacity region was characterized for the
strong interference regime. A simpler variation of the Gaussian $K$-user CZIC has been studied in \cite{ErkipCIC}, where the
results of \cite{WeiYuCIC} are strengthened for the $3$-user case. It is shown in \cite{ErkipCIC} that a generalization of
the Han-Kobayashi \cite{Han:1981} scheme can achieve sum-capacity for some interference regimes.

In this paper we focus on the $K$-user CZIC \emph{with feedback}, i.e., we assume the presence of noiseless and causal feedback from the $k$th receiver to the $k$th transmitter. For $K=2$, this model reduces to the conventional $2$-user interference channel with feedback. For $K>2$, this model is a
special case of the general $K$-user interference channel with feedback (see Figure \ref{figure1model}).
The $2$-user interference channel with various forms of feedback has been investigated recently. Feedback coding schemes for $K$-user
Gaussian interference networks have been developed by Kramer in \cite{KramerIC:2002}. Outer bounds for the $2$-user interference channel
with generalized feedback have been derived in \cite{GK3:2006}, \cite{TuninettiOuter} and \cite{TandonUlukus:2011} (also see references therein).
The $2$-user Gaussian interference channel with noiseless (channel output) feedback was considered in \cite{SuhTseIT} and  the feedback
capacity region was characterized to within two bits. One of the main findings in \cite{SuhTseIT} is that feedback provides
multiplicative gain at high signal-to-noise ratio ($\SNR$) and the gain can be arbitrarily large for certain channel parameters.
The key insights that led to this result were obtained by characterizing the feedback capacity region of the linear
deterministic (LD) $2$-user interference channel. The linear deterministic model despite its simplicity can provide valuable
insights for the Gaussian channel model.
\begin{figure}[t]
\centering
\includegraphics[width=0.45\textwidth]{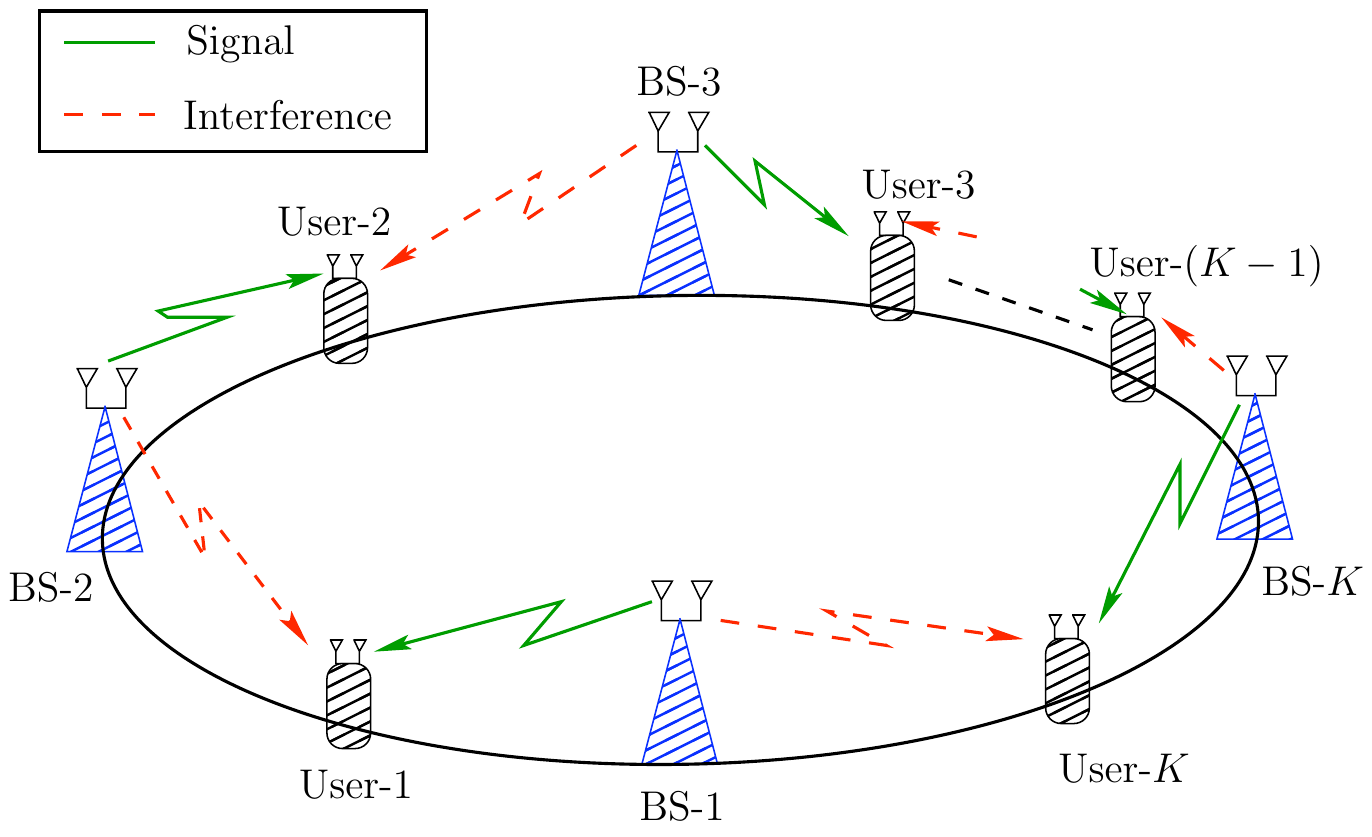}
  \caption{Modified Wyner Model.}\label{figurewynermodel}
  \vspace{-0.5cm}
\end{figure}

With this correspondence at hand, we first focus on the linear deterministic $K$-user
CZIC with feedback. We characterize the symmetric feedback capacity, $\mathcal{C}_{\mathrm{sym},\mathrm{LD}}^{\mathrm{FB}}$, which is
defined as the maximum $R$ such that the rate $K$-tuple
$(R,R,\ldots,R)$ is achievable with feedback. We use insights from the linear deterministic model to characterize the
symmetric feedback capacity of the $K$-user Gaussian CZIC within a constant number of bits (independent of the channel gains) for all
interference regimes. As a consequence of our constant gap results, we also establish the generalized  degrees of freedom of the Gaussian CZIC with feedback.
For the scope of this paper, we restrict our attention to the case of symmetric channel parameters. For instance, for the Gaussian CZIC
with feedback, we assume that the direct channel gain from the $k$th transmitter to the $k$th receiver is the same for all $k$, and that the interference
channel gain from the $k$th transmitter to the $(k-1)\hspace{-0.1in}\mod(K)$th receiver is the same for all $k$.

The symmetric $K$-user Gaussian CZIC is described by the pair $(\SNR,\INR)$, where $\SNR$ denotes the direct channel gain and $\INR$ denotes the
interference channel gain. The per-user generalized degrees of freedom ($\DOF$) of the Gaussian CZIC \emph{without feedback} is defined as
\begin{align}
{\DOF}(\alpha,K)&= \frac{1}{K}\lim_{\SNR\rightarrow \infty}\frac{\mathcal{C}_{\mathrm{sum},\mathrm{G}}(K)}{\frac{1}{2}\log(1+\SNR)},\label{DOFdefn1}
\end{align}
where $\mathcal{C}_{\mathrm{sum},\mathrm{G}}(K)$ is the sum-capacity \emph{without feedback}, and  $\alpha$ is the interference parameter, defined as  $\alpha\triangleq \frac{\log(\INR)}{\log(\SNR)}$. Analogously to (\ref{DOFdefn1}), we define the per-user generalized degrees of freedom of the Gaussian CZIC \emph{with feedback} as
\begin{align}
{\DOF^{\mathrm{FB}}}(\alpha,K)&= \frac{1}{K}\lim_{\SNR\rightarrow \infty}\frac{\mathcal{C}_{\mathrm{sum},\mathrm{G}}^{\mathrm{FB}}(K)}{\frac{1}{2}\log(1+\SNR)},
\end{align}
where $\mathcal{C}_{\mathrm{sum},\mathrm{G}}^{\mathrm{FB}}(K)$ is the sum-capacity \emph{with feedback}.

\begin{figure}[t]
\centering
\includegraphics[width=0.38\textwidth]{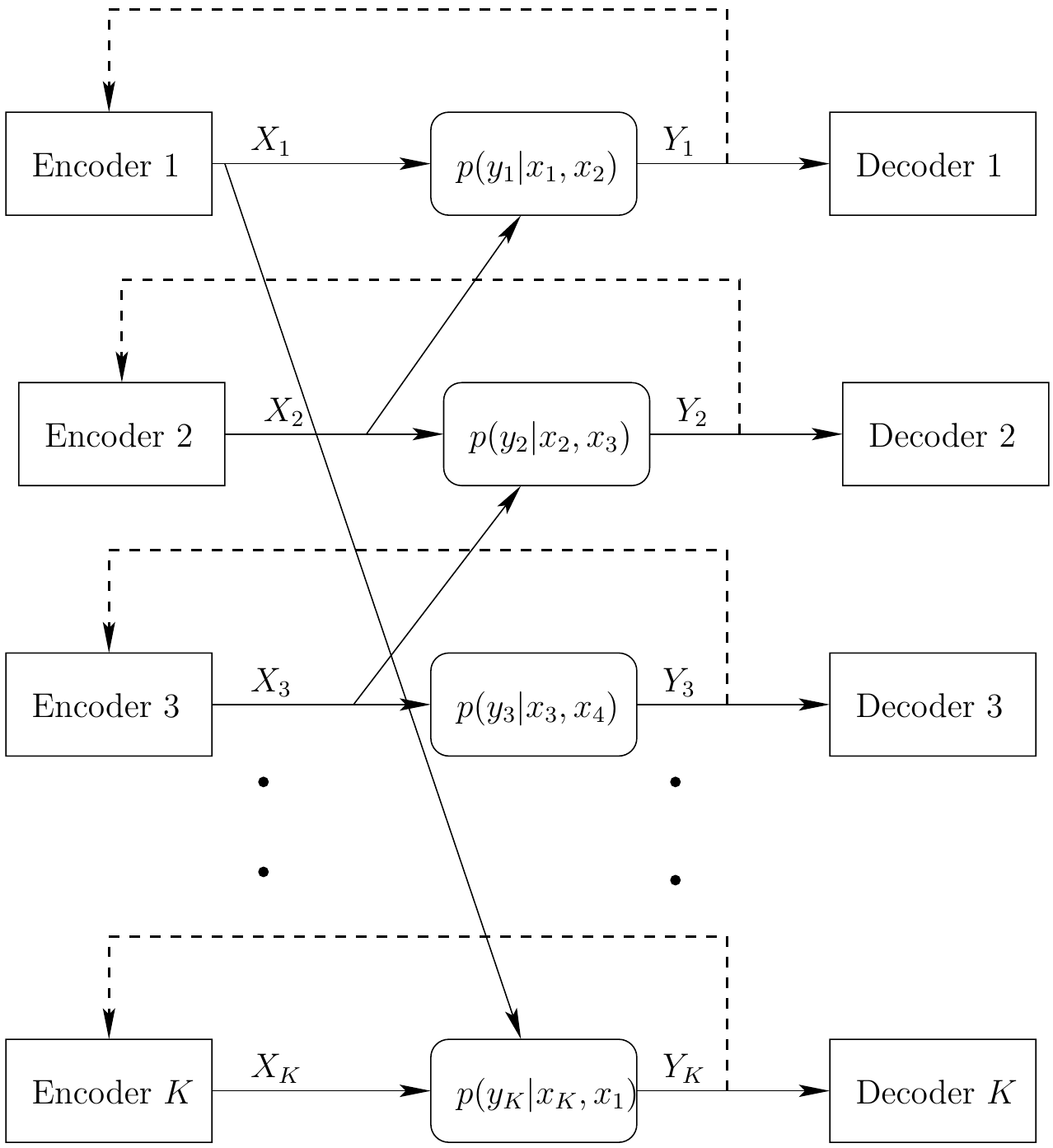}
  \caption{$K$-user cyclic Z-interference channel with feedback.}\label{figure1model}
  \vspace{-0.6cm}
\end{figure}
In the breakthrough work \cite{ETW2008, Salman2011}, several novel results were obtained for the $2$-user Gaussian interference channel. Among the results, is the
characterization of the $\DOF$. In particular, it was shown that $\DOF(\alpha,2)= \min(\max(1-\alpha,\alpha),1-\alpha/2)$.
On the other hand, the feedback $\DOF$ for $K=2$ was recently characterized in \cite{SuhTseIT} and is given as ${\DOF^{\mathrm{FB}}}(\alpha,2)=\max(1-\alpha/2,\alpha/2)$.
From the results of \cite{WeiYuCIC}, it is clear that the $\DOF$ without feedback for the
$K$-user Gaussian CZIC is the same for all $K\geq 2$, i.e.,
it is \emph{independent} of $K$, the number of users. It is natural to ask whether this equivalence
continues to hold in presence of feedback for $K>2$.

We answer this question in the negative by showing that
the feedback $\DOF$ of the $K$-user Gaussian CZIC is in general a function of $K$. Of particular interest
is the very strong interference regime, corresponding to $\alpha\geq 2$. In this regime the feedback
$\DOF$ for the $2$-user case is given as $\DOF^{\mathrm{FB}}(\alpha,2)=\alpha/2$.
This implies that the feedback gain can be unbounded as $\alpha$ increases.
For this regime, we show that the feedback $\DOF$ of the $K$-user Gaussian CZIC is given as $\DOF^{\mathrm{FB}}(\alpha,K)=1+ \frac{(\alpha-2)}{K}$.
This result shows that for a fixed $\alpha$, as the number of users increases, the feedback gain decreases and completely
vanishes in the limit $K\rightarrow \infty$. The outer bounds derived in this paper to establish capacity/constant bit gap results can be regarded as
genie aided bounds derived for the $2$-user case considered in \cite{SuhTseIT}.
However, as $K$, the number of users increases, selecting appropriate genies becomes prohibitively complex.
In particular, for the $K$-user CZIC, we have a total of $K!$ sum-rate upper bounds. Depending on the interference
parameter $\alpha$, we carefully select the best upper bound among the $K!$ upper bounds.

This paper is organized as follows. In Section \ref{Section:MODEL}, we describe the $K$-user cyclic Z-interference channel with feedback. In Section \ref{Section:mainresults} we describe our main results for
both linear deterministic and Gaussian $K$-user CZICs. We provide intuition as to why the feedback gain decreases
as the number of users increases. Proofs for the $K$-user linear deterministic CZIC are presented in Sections \ref{Section:Coding} and
\ref{Section:Upper}. Constant gap results for the feedback capacity of the $K$-user Gaussian CZIC are established in Section \ref{GaussianSection}. We conclude the paper in Section \ref{CONCLUSION}. Parts of this paper have appeared in \cite{TP:Allerton2011}.

\section{$K$-user Cyclic Z-IC with Feedback}\label{Section:MODEL}
The $K$-user cyclic Z-interference channel is described by $K$
conditional probabilities $\{p(y_{1}|x_{1},x_{2}),
\ldots,p(y_{K}|x_{K},x_{1})\}$. A
$(T,M_{1},\ldots,M_{K})$ feedback code for the CZIC consists of sequences of $K$ encoding functions
\begin{align}
f_{k,t}:\{1,\ldots,M_{k}\}\times\mathcal{Y}^{t-1}_{k}\rightarrow
\mathcal{X}_{k,t},\quad k=1,\ldots,K,\label{encodingfunction}
\end{align}
for $t= 1,\ldots,T$, and $K$ decoding functions
\begin{align}
g_{k}: \mathcal{Y}_{k}^{T}\rightarrow \{1,\ldots,M_{k}\},\quad k=1,\ldots,K.
\end{align}
 The probability of decoding error at decoder $k$ is denoted by $P_{k}$ and is defined as
$P_{k}=\mathbb{P}(g_{k}(Y_{k}^{T})\neq W_{k})$, where $W_{k}$ is the message of transmitter $k$.

A rate $K$-tuple $(R_{1},\ldots,R_{K})$ is
achievable for the $K$-user CZIC if there exists a \break
$(T,M_{1},\ldots,M_{K})$ feedback code such that $\log(M_{k})/T\leq
R_{k}-\epsilon_{k,T}$ and $P_{k}\leq \epsilon_{k,T}$, where
$\epsilon_{k,T}\rightarrow 0$ as $T\rightarrow \infty$ for all $k$. The feedback capacity region $\mathcal{C}^{\mathrm{FB}}(K)$
is the set of all achievable $K$-tuples.

In this paper, we focus on the symmetric feedback capacity of the $K$-user CZIC, denoted by $\mathcal{C}_{\mathrm{sym}}^{\mathrm{FB}}(K)$, which is defined as the maximum
$R$ such that $(R,\ldots,R)\in \mathcal{C}^{\mathrm{FB}}(K)$.

\subsection{Linear deterministic CZIC with Feedback}\label{Section:LDCZIC}
The symmetric linear deterministic CZIC is described by a pair of integers $(n,m)$,
where $n$ denotes the number of signal (direct) levels and $m$ denotes the number of interference
levels observed at the receivers.

The channel input of transmitter $k$, denoted by $X_{k}$, for $k=1,\ldots,K$, is assumed to be of length $\max(n,m)$.

When $n\geq m$, we denote
\begin{align}
U_{k}&: \mbox{ top-most } (n-m) \mbox{ bits of } X_{k}\nonumber\\
V_{k}&: \mbox{ top-most } m \mbox{ bits of } X_{k}\label{notationweak}\\
L_{k}&: \mbox{ lower-most } m \mbox{ bits of } X_{k}\nonumber.
\end{align}
With this notation, we can write the channel outputs for the $K$-user LD-CZIC as follows:
\begin{align}
Y_{k}&=(U_{k}, L_{k}\oplus V_{k+1}),
\end{align}
for $k=1,\ldots,K$.

When $n< m$, we denote
\begin{align}
U_{k}&: \mbox{ top-most } (m-n) \mbox{ bits of } X_{k}\nonumber\\
V_{k}&: \mbox{ top-most } n \mbox{ bits of } X_{k}\label{notationstrong}\\
L_{k}&: \mbox{ lower-most } n \mbox{ bits of } X_{k}\nonumber.
\end{align}
With this notation, we can write the channel outputs for the $K$-user LD-CZIC as follows:
\begin{align}
Y_{k}&=(U_{k+1}, L_{k+1}\oplus V_{k}),\hspace{0.1in}
\end{align}
for $k=1,\ldots,K$, where we define
\begin{align}
V_{K+1}\triangleq V_{1},\quad U_{K+1}\triangleq U_{1},\quad L_{K+1}\triangleq L_{1}
\end{align}
for consistency.

For instance, when $n\geq m$, the $3$-user LD-CZIC is described by the following input-output relationships:
\begin{align}
Y_{1}&=(U_{1},L_{1}\oplus V_{2})\nonumber\\
Y_{2}&=(U_{2},L_{2}\oplus V_{3})\nonumber\\
Y_{3}&=(U_{3},L_{3}\oplus V_{1})\nonumber.
\end{align}

\subsection{Gaussian $K$-user CZIC with Feedback}
To describe the $K$-user Gaussian CZIC, we denote\footnote{With slight abuse of notation, we use similar notation for both LD-CZIC and G-CZIC channel models. However, the corresponding notation should be clear from the context.} the signal transmitted by user $k$ as $X_k$.
We impose an average unit power constraint at each user; that is $\E[X_k^2]\leq 1$. The signal observed at receiver $k$ is obtained by
\begin{align}
Y_k= \sqrt{\SNR} X_k+ \sqrt{\INR} X_{k+1} +Z_k, \qquad k=1,2,\dots,K,
\end{align}
where we define $X_{K+1} \triangleq X_1$ for consistency, and the noise $Z_k$ at receiver $k$ is a zero mean Gaussian random variable with unit variance.
Moreover, the noises across the receivers, i.e., $Z_k$ and $Z_{k^{'}}$ are independent for $k\neq k^{'}$.

\section{Main Results}\label{Section:mainresults}
The results for the linear deterministic model are presented
in terms of the interference parameter $\alpha$, which is defined in this model
as the ratio of the number of interference levels to the number of signal levels, i.e.,
\begin{align}
\alpha&\triangleq \frac{m}{n}.
\end{align}
We define the normalized\footnote{The normalization is with respect to the number of direct levels, $n$.} symmetric feedback capacity {\emph{per-user} of the LD-CZIC as follows:
\begin{align}
\mathcal{C}^{\mathrm{FB}}_{\mathrm{sym},\mathrm{LD}}(\alpha,K)&\triangleq \frac{1}{K}\frac{\mathcal{C}^{\mathrm{FB}}_{\mathrm{sum},\mathrm{LD}}(K)}{n},
\end{align}
where $\mathcal{C}^{\mathrm{FB}}_{\mathrm{sum},\mathrm{LD}}(K)$ is the feedback sum-capacity of the $K$-user LD-CZIC.

We present our first result in the following theorem:
\begin{Theo}\label{Theorem1}
The normalized symmetric feedback capacity, $\mathcal{C}^{\mathrm{FB}}_{\mathrm{sym},\mathrm{LD}}(\alpha,K)$
of the $K$-user LD-CZIC is given by
\begin{align}
\mathcal{C}^{\mathrm{FB}}_{\mathrm{sym},\mathrm{LD}}(\alpha,K)=
\begin{cases}
(1-\alpha)+\frac{\alpha}{K}, &0\leq\alpha\leq 1/2\\
\alpha + \frac{(2-3\alpha)}{K}, &1/2\leq \alpha\leq 2/3\\
1-\frac{\alpha}{2}, &2/3\leq \alpha\leq 1\\
\frac{\alpha}{2}, &1\leq \alpha\leq 2\\
1+ \frac{(\alpha-2)}{K},&\alpha\geq 2.
\end{cases}
\end{align}
\end{Theo}
\vspace{0.1in}
Theorem \ref{Theorem1} is proved in two parts: feedback coding schemes are presented in Section \ref{Section:Coding} and corresponding
upper bounds for the normalized symmetric feedback capacity are
obtained in Section \ref{Section:Upper}.

The constant bit gap results for the Gaussian model are presented in terms of two parameters $(C_{\SNR},C_{\INR})$, defined as
\begin{align}
C_{\SNR}&\triangleq \frac{1}{2}\log(1+\SNR)\\
C_{\INR}&\triangleq \frac{1}{2}\log(1+\INR).
\end{align}
We also define the interference parameter for the Gaussian model as
\begin{align}
\alpha&=\frac{\log(\INR)}{\log(\SNR)}.
\end{align}

We further define the symmetric feedback capacity {\emph{per-user}} of the $K$-user Gaussian CZIC as follows:
\begin{align}
\mathcal{C}^{\mathrm{FB}}_{\mathrm{sym},\mathrm{G}}(K)&\triangleq \frac{\mathcal{C}^{\mathrm{FB}}_{\mathrm{sum},\mathrm{G}}(K)}{K},
\end{align}
where $\mathcal{C}^{\mathrm{FB}}_{\mathrm{sum},\mathrm{G}}(K)$ is the feedback sum-capacity of the $K$-user Gaussian CZIC.
We next define the feedback degrees of freedom {\emph{per-user}} for the $K$-user Gaussian CZIC as follows:
\begin{align}
\DOF^{\mathrm{FB}}(\alpha,K)&= \lim_{\SNR\rightarrow \infty}\frac{\mathcal{C}_{\mathrm{sym},\mathrm{G}}^{\mathrm{FB}}(K)}{\frac{1}{2}\log(1+\SNR)}.
\end{align}

We present our next result in the following theorem:
\begin{Theo}\label{TheoremGaussiangap}
The symmetric feedback capacity per user, $\mathcal{C}^{\mathrm{FB}}_{\mathrm{sym},\mathrm{G}}(K)$
of the $K$-user Gaussian CZIC satisfies
\begin{align}
\mathcal{C}^{\mathrm{FB}}_{\mathrm{sym},\mathrm{G}}(K)\backsimeq
\begin{cases}
C_{\SNR}-C_{\INR}+\frac{C_{\INR}}{K}, &0\leq\alpha\leq 1/2\\
C_{\INR} + \frac{(2C_{\SNR}-3C_{\INR})}{K}, &1/2\leq \alpha\leq 2/3\\
C_{\SNR}-\frac{C_{\INR}}{2}, &2/3\leq \alpha\leq 1\\
\frac{C_{\INR}}{2}, &1\leq \alpha\leq 2\\
C_{\SNR}+ \frac{(C_{\SNR}-2C_{\INR})}{K},&\alpha\geq 2,
\end{cases}
\end{align}
where the notation $A\backsimeq B$ implies that $(A-B)\leq 3$, i.e., the worst case gap (for all interference regimes) between the upper and lower bounds is at most $3$ bits/user.
\end{Theo}
Theorem \ref{TheoremGaussiangap} is proved in Section \ref{GaussianSection}, where we use key insights from the
linear deterministic model to construct feedback coding schemes and corresponding upper bounds on the feedback sum capacity.
Further analysis of these bounds shows that they differ by a constant number of bits, which is independent of $(\SNR,\INR)$.
We note here that the worst case gap of $3$ can be reduced depending on the interference regime.
For instance, in our proof of Theorem \ref{TheoremGaussiangap}, we show that the gap for the case when $\alpha>2$ is at most $2$ bits.

As a consequence of Theorem \ref{TheoremGaussiangap}, we have the following corollary:

\begin{Cor}\label{corollary1}
The feedback degrees of freedom per-user of the $K$-user Gaussian CZIC is given by
\begin{align}
\DOF^{\mathrm{FB}}(\alpha,K)=
\begin{cases}
(1-\alpha)+\frac{\alpha}{K}, &0\leq \alpha\leq 1/2;\\
\alpha + \frac{(2-3\alpha)}{K}, &1/2\leq \alpha\leq 2/3;\\
1-\frac{\alpha}{2}, &2/3\leq \alpha\leq 1;\\
\frac{\alpha}{2}, &1\leq \alpha\leq 2;\\
1+ \frac{(\alpha-2)}{K}, &\alpha\geq 2.
\end{cases}
\end{align}
\end{Cor}

We recall the no-feedback degrees of freedom per-user of the $K$-user Gaussian CZIC \cite{WeiYuCIC}:
\begin{align}\label{nofeedback}
\DOF(\alpha,K)=
\begin{cases}
(1-\alpha), &0\leq\alpha\leq 1/2\\
\alpha, &1/2\leq \alpha\leq 2/3\\
1-\frac{\alpha}{2}, &2/3\leq \alpha\leq 1\\
\frac{\alpha}{2}, &1\leq \alpha\leq 2\\
1,&2\leq\alpha.
\end{cases}
\end{align}
Note that $\DOF(\alpha,K)$ is \emph{independent} of $K$,
i.e. $\DOF(\alpha,K)=\DOF(\alpha,2)$, for all $K$.
This implies that from the $\DOF$ point of view, the behavior of the $K$-user
system is similar to the $K=2$ user system in the \emph{absence} of feedback.
\begin{figure}[t]
\centering
\includegraphics[width=0.52\textwidth]{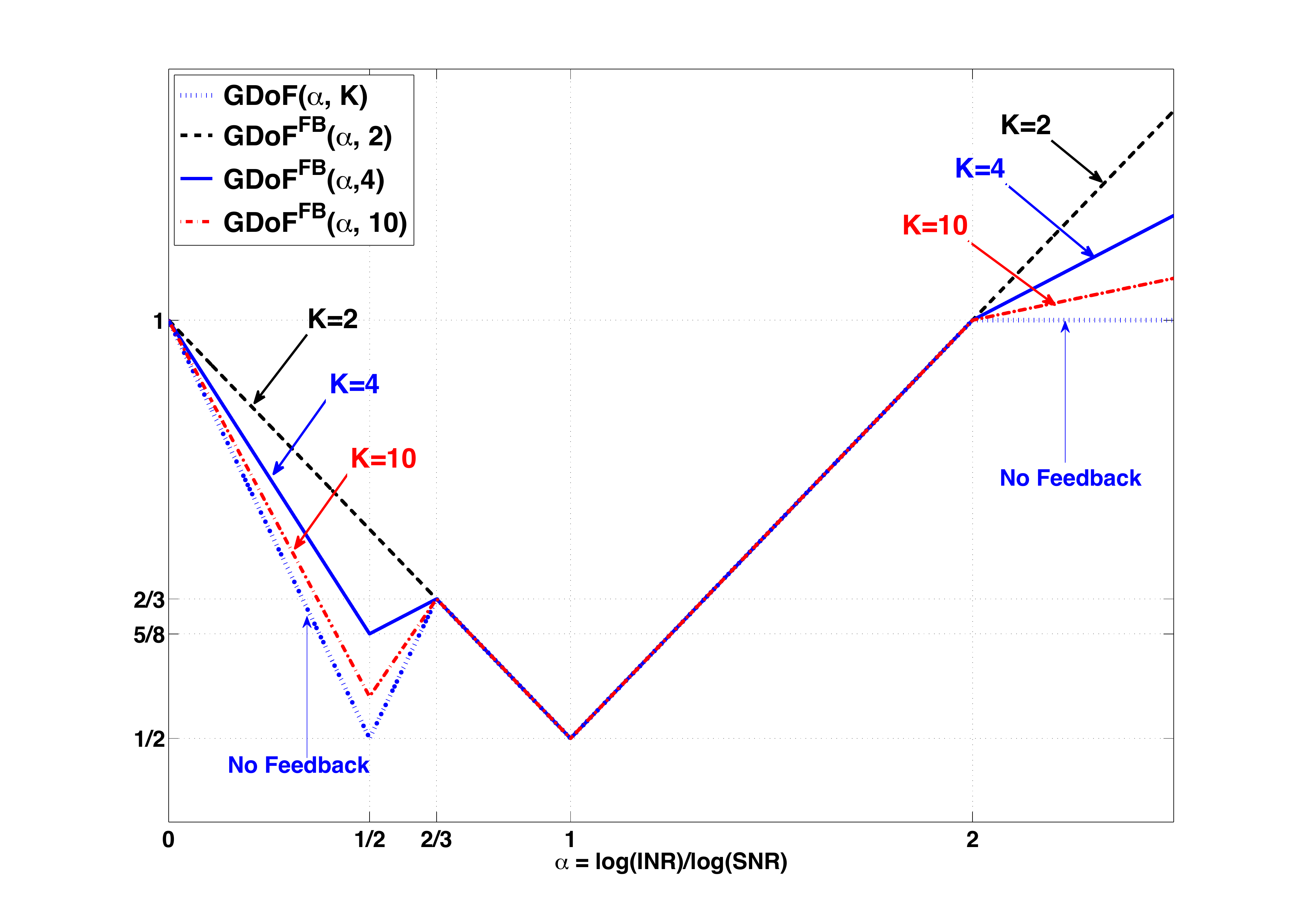}
  \vspace{-0.7cm}
  \caption{ $\DOF$ of the $K$-user Gaussian CZIC with and without feedback.}\label{figurecapacityK}
  \vspace{-0.6cm}
\end{figure}

On the other hand, we note that the feedback $\DOF$ for $K=2$ is given by \cite{SuhTseIT}
\begin{align}
\DOF^{\mathrm{FB}}(\alpha,2)=
\begin{cases}
1-\frac{\alpha}{2}, &\alpha\leq 1\\
\frac{\alpha}{2}, &\alpha\geq 1.
\end{cases}
\end{align}
In the light of above observations, it is natural to ask whether the
behavior of the $K$-user Gaussian CZIC mimics the behavior of the $K=2$ system in
the presence of feedback. Corollary \ref{corollary1} answers this question in the negative by
showing that the feedback $\DOF$ per-user for $K>2$ is in general a function of $K$. Moreover, the feedback $\DOF$ of
$K=2$ always serves as an upper bound for the feedback $\DOF$ for $K>2$ users.

In Figure \ref{figurecapacityK}, the feedback $\DOF$s are shown for the $K$-user Gaussian CZIC, when $K=2,4$ and $10$.

\begin{Rem}
Corollary \ref{corollary1} also shows that $\DOF^{\mathrm{FB}}(\alpha,K)$ can
be strictly less than $\DOF^{\mathrm{FB}}(\alpha,2)$ (see Figure \ref{figurecapacityK}).
Secondly, it also shows that $\DOF^{\mathrm{FB}}(\alpha,K)$ is monotonically decreasing in
$K$. Hence, as the number of users in the system increases,
the $\DOF$ gain obtained via feedback decreases. Furthermore, in the limit
$K\rightarrow \infty$, the feedback gain vanishes, i.e., we have
\begin{align}
\lim_{K\rightarrow \infty}\DOF^{\mathrm{FB}}(\alpha,K)= \DOF(\alpha,K).
\end{align}
\end{Rem}

The results presented so far show that the capacity gain provided by feedback decreases as $K$ increases. We should remark here that
this behavior of capacity is not a universal phenomenon and is not
necessarily dependent on the cyclic network topology. We show in the next section that this phenomenon is an artifact of the \emph{local}
feedback assumption. Under the local feedback assumption, receiver $k$ feeds
back its channel output \emph{only} to transmitter $k$. To avoid any confusion, whenever we refer to feedback, it should be clear that we
are referring to  \emph{local} feedback.
\begin{figure*}[t]
\centering
\includegraphics[width=0.65\textwidth]{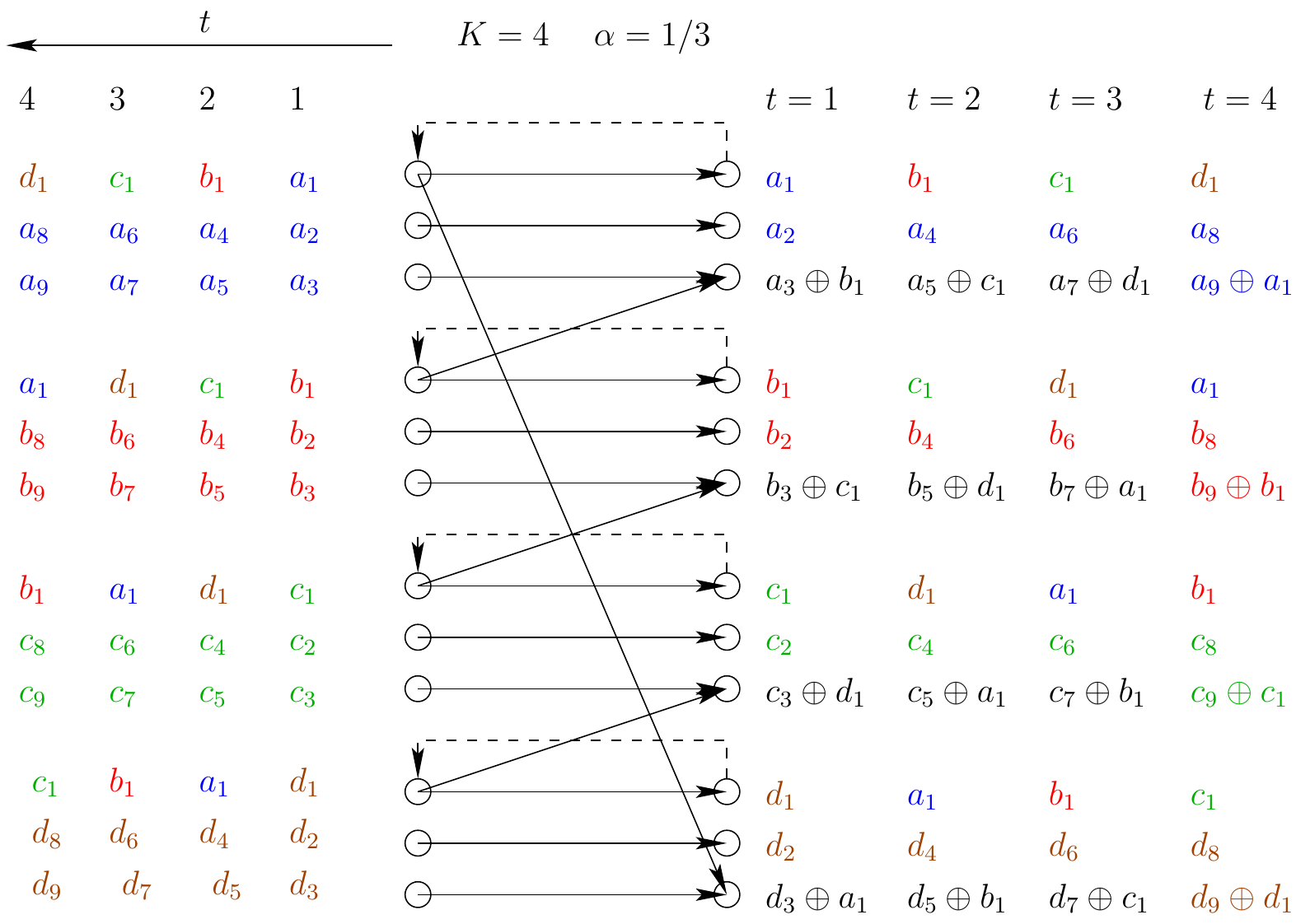}
  \caption{Feedback coding scheme for $\alpha=1/3$, $K=4$.}\label{SchemeKevenveryweak}
  \vspace{-0.5cm}
\end{figure*}

On the contrary,  under the stronger model of global feedback, i.e., a
model in which all receivers feed their channel outputs back to all the
transmitters, the feedback gain is \emph{independent} of $K$. We present the sum
capacity of the LD-CZIC with global feedback in the following theorem.
\begin{Theo}\label{theoremGlobalFB}
The normalized symmetric feedback capacity of the $K$-user LD-CZIC
with global feedback is given by
\begin{align}
\mathcal{C}^{\mathrm{FB}}_{\mathrm{sym},\mathrm{LD}}(\alpha,K)&= \max\left(1-\frac{\alpha}{2},\frac{\alpha}{2}\right).
\end{align}
\end{Theo}
We present the coding scheme with global feedback in Section \ref{coding:globalFB} and mention the converse after Theorem \ref{theoremK2}.
The setting of symmetric $K$-user fully connected LD-IC (and extensions to the Gaussian model) have also been considered in the literature and the interested reader is referred to \cite{MTP:ITJournal}.

\section{Feedback Coding Schemes for LD-CZIC}\label{Section:Coding}
\subsection{Very-weak interference: $0\leq \alpha\leq 1/2$}
\label{sec:LD-VW}
In this regime, we show that $K(n-m)+ m$ bits per user can be reliably sent
in $K$ channel uses.

As an example, we start with the case in which $K=4$,  $m=1$ and $n=3$, so
that $\alpha=1/3$. To achieve $9$ bits per user in $4$ channel uses, the
following coding scheme is used (see Figure \ref{SchemeKevenveryweak}):

\begin{itemize}
\item In the first channel use, each encoder transmits fresh bits on all
  levels (for example, encoder $1$ sends $a_{1},a_{2},a_{3}$, encoder $2$ sends $b_{1}, b_{2}, b_{3}$ etc.).

\item Upon receiving feedback, each encoder can decode the upper most
  bit of the next encoder (encoder $1$ decodes $b_{1}$, encoder $2$
   decodes $c_{1}$, etc.).

\item In all subsequent channel uses, each encoder
   transmits the  previously decoded bit on the top most level and
fresh information bits in the remaining lower two levels (at $t=2$ encoder
   $1$ transmits $b_{1}$ on the top level and $a_{4},a_{5}$ on the two
   lower levels).

\item  From Figure \ref{SchemeKevenveryweak} it is clear that each
  user can reliably transmit $9$ bits to its decoder in $4$ channel
  uses.  Hence, this scheme yields a normalized symmetric rate of
  $(9/4)\times (1/3)= 3/4$.
\end{itemize}

This scheme can be readily generalized for \emph{arbitrary} numbers of users
$K$ and for any $\alpha\in [0,1/2]$ as follows: each transmit signal can split its signal into
three mutually exclusive sets of levels as $X_k(t)=(X_{k,1}(t),X_{k,2}(t),X_{k,3}(t))$, where the number of bits
in $X_{k,1}(t)$, $X_{k,2}(t)$, and $X_{k,3}(t)$ are $m$, $(n-2m)$, and $m$ respectively. At $t=1$, each
encoder transmits $n$ fresh information bits over all its transmit levels.  Using
feedback, at the end of slot time $t-1$ ($1<t\leq K$), the
$k$th encoder decodes the upper most $m$ bits transmitted by the $(k+1)$th
encoder, that is $X_{k+1,1}(t-1)$. For all the subsequent channel uses
the $k$th encoder transmits the previously decoded $m$ bits on its top
most $m$ levels, and transmits fresh information in the lower $(n-m)$
levels. More precisely, $X_{k,1}(t)=X_{k+1,1}(t-1)$, and $X_{k,2}(t)$ and  $X_{k,3}(t)$ consist of $(n-m)$
fresh symbols intended for the $k$th receiver. \\
At the end time slot $t$, the $k$th receiver can decode $X_{k,1}(t)$ and $X_{k,2}(t)$ which are received cleanly. The remaining part of the received signal would be $X_{k,3}(t)\oplus X_{k+1,1}(t)$, $X_{k,3}(t)$, and so $X_{k,3}(t)$ cannot be decoded right away because it is corrupted by interference. However, note that this interference will be retransmitted by the $k$th encoder in the next time slot. Hence,  once $X_{k,1}(t+1)=X_{k+1,1}(t)$ is decoded in the next block, receiver $k$ can remove it from the interfered signal, and determine $X_{k,3}(t)$.
This scheme achieves $K(n-m) + m$ bits per user in
$K$ channel uses and the achievable rate is $(n-m)+m/K$. Hence for $\alpha\in [0,1/2]$, we have
\begin{align}
\mathcal{C}^{\mathrm{FB}}_{\mathrm{sym},\mathrm{LD}}(\alpha,K)&\geq (1-\alpha) + \frac{\alpha}{K}.
\end{align}

\subsection{Weak interference: $1/2\leq \alpha\leq 2/3$}
For this regime, we present a feedback coding scheme that achieves $Km+ (2n-3m)$ bits
per user in $K$ channel uses. We break the channel input of encoder $k$ into four mutually exclusive sets of levels as follows: $X_{k}(t)=(X_{k,1}(t),X_{k,2}(t),X_{k,3}(t),X_{k,4}(t))$,
where the number of bits in $X_{k,r}(t)$ are $(2m-n),(2n-3m),(2m-n)$ and $(n-m)$, for $r=1,2,3$ and $4$, respectively. Note that condition 
$1/2 \leq \alpha \leq 2/3$ ensures that the number of levels in each partition is non-negative. 
\begin{figure*}[t]
\centering
\includegraphics[width=0.65\textwidth]{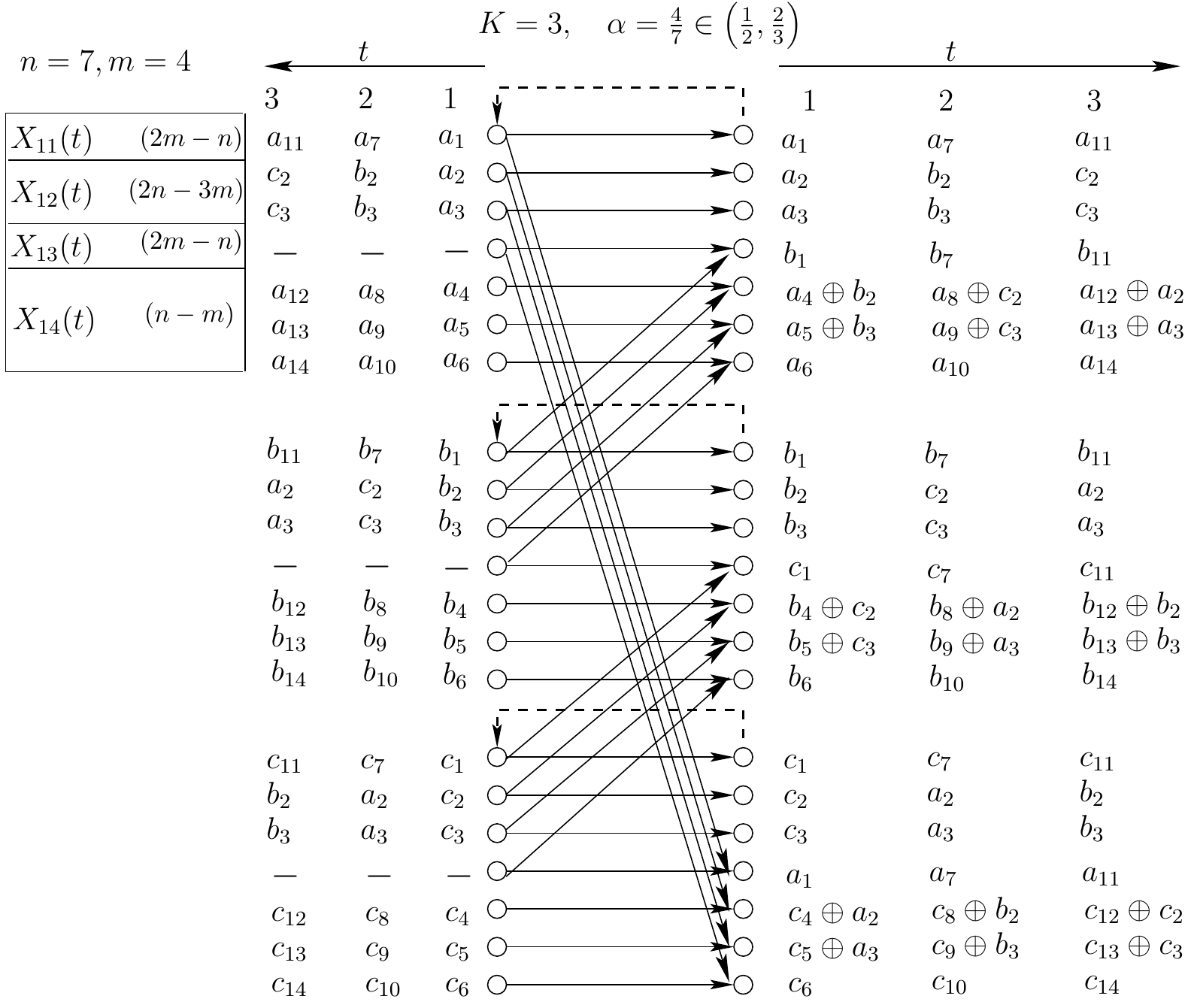}
  \caption{Feedback coding scheme for $\alpha=4/7$, $K=3$.}\label{Schemeweak}
\vspace{-0.5cm}
\end{figure*}

At $t=1$, each encoder transmits fresh information bits on $X_{k,1}(1),X_{k,2}(1)$ and $X_{k,4}(1)$ levels.
For all $t\in \{1,\ldots,K\}$, all encoders remains silent in the $X_{k,3}(t)$ level. At $t$, due to feedback, encoder $k$
can decode the bits transmitted by the encoder $(k+1)$ over its top $m$ levels, including $X_{(k+2),2}$, since $(2m-n) + (2n-3m) \leq m$. . For all $1<t\leq K$, encoder $k$ transmits
\begin{align}
X_{k}(t)=(X_{k,1}(t),X_{(k+1),2}(t-1),\phi,X_{k,4}(t))\nonumber,
\end{align}
where $X_{k,1}(t)$ and $X_{k,4}(t)$ consist of fresh information bits. It is clear that $Km$ bits are achievable from the levels $1$ and $4$. A gain of $(2n-3m)$ bits is provided by feedback in $K$ uses of the channel. It can be easily verified that this coding scheme yields $Km+ (2n-3m)$ bits
per user in $K$ channel uses.
Hence, we have
\begin{align}
\mathcal{C}^{\mathrm{FB}}_{\mathrm{sym},\mathrm{LD}}(\alpha,K)&\geq \alpha + \frac{(2-3\alpha)}{K}.
\end{align}
This scheme is illustrated through an example in Figure \ref{Schemeweak} for the case in which $K=3$ and $n=7,m=4$, so that $\alpha=4/7$.
The partition of the input at each user into four mutually exclusive levels can be readily understood through this example.
It is clear that for these parameters, the scheme achieves $Km+ (2n-3m) =14$ bits per user in $K=3$ channel uses.

\subsection{Moderate-strong interference: $2/3\leq \alpha\leq 2$}
In this regime, Theorem \ref{Theorem1} shows that feedback does not increase the normalized symmetric capacity and
hence the no-feedback coding scheme in \cite{WeiYuCIC} suffices.

\subsection{Very-strong interference: $ \alpha\geq 2$}
In this regime, we will show that $(K-2)n+m$ bits per user are
achievable in $K$ channel uses.

As an example, we start with the case in which $K=4$, $m=3$ and $n=1$, so
that $\alpha=3$. To achieve $5$ bits per user in $4$ channel uses, the
following coding scheme is used (see Figure \ref{SchemeKeven3}):

\begin{itemize}

\item  In all channel uses, each encoder remains silent in the lower-most
bit. In the first channel use, each encoder transmits $2$ fresh
bits (for instance, encoder $1$ sends $a_{1},a_{2}$ and encoder $2$ sends
$b_{1},b_{2}$). Using feedback, each encoder can decode the second bit
transmitted by the encoder interfering with its decoder (encoder $1$
decodes $b_{2}$, encoder $2$ decodes $c_{2}$, etc.).

\item In all subsequent channel uses, each encoder transmits a fresh
information bit in the top-most level and the previously decoded bit
in the second level (for instance, at $t=2$, encoder $2$ sends the
fresh bit $b_{3}$ in the top-most level and the decoded bit $c_{2}$ in
the second level).

\item From Figure \ref{SchemeKeven3}, it is clear that in
$4$ channel uses, using the top-most level, each decoder receives $4$
bits. One more bit is received from the interfering user in the final
channel use (for instance, the bit $a_{2}$ is received at decoder $1$
in a delayed manner). This scheme yields a rate of $5/4$ per
user.
\end{itemize}
\begin{figure*}[t]
\centering
\includegraphics[width=0.65\textwidth]{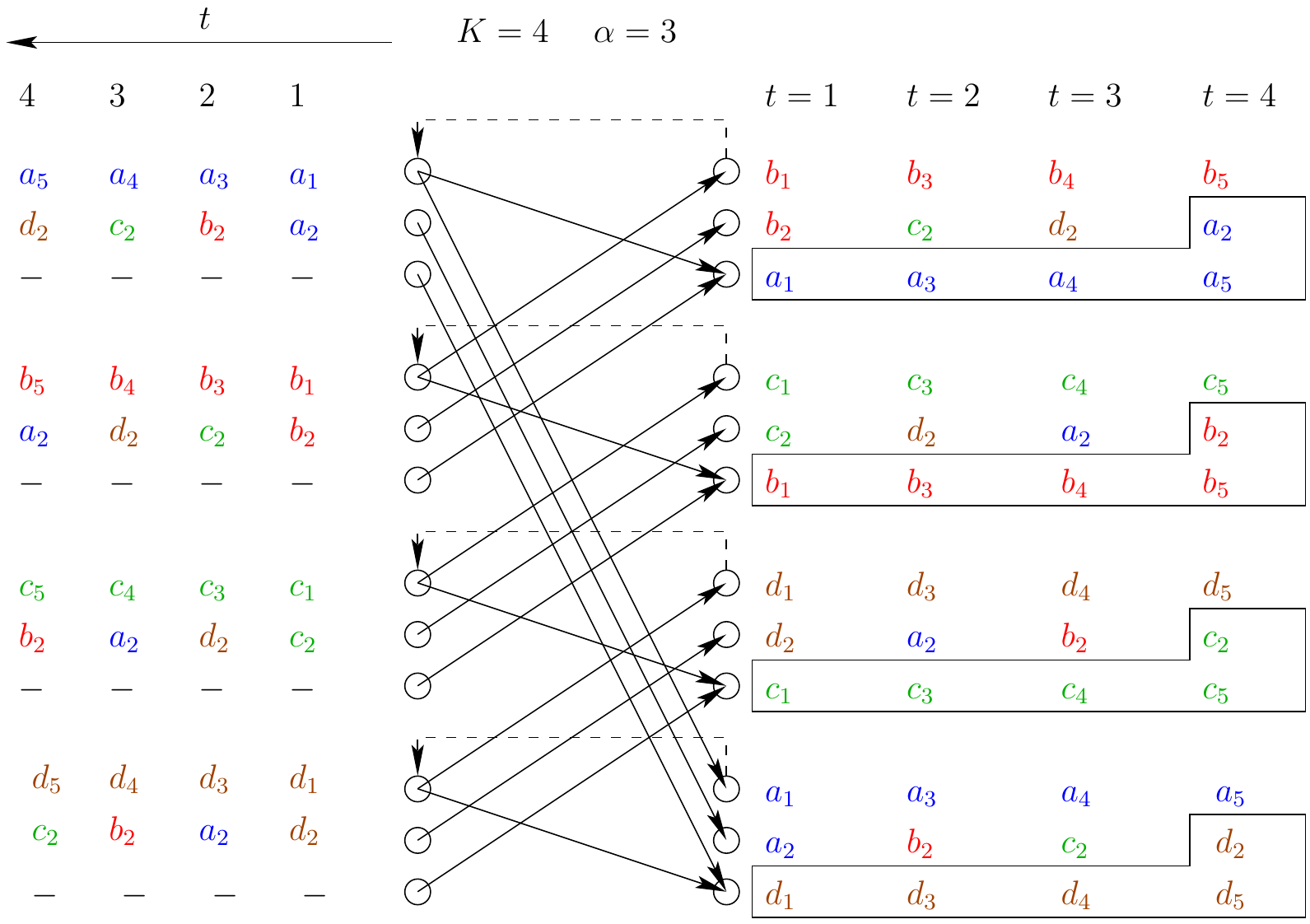}
  \caption{Feedback coding scheme for $\alpha=3$, $K=4$.}\label{SchemeKeven3}
\vspace{-0.5cm}
\end{figure*}

This scheme can be readily generalized for \emph{arbitrary} numbers of users
$K$ and for any $\alpha\geq 2$ as follows: for any $1\leq t\leq K$,
all encoders  transmit no information in the lower-most $n$ levels.
At $t=1$, the $k$th encoder transmits $(m-n)$ fresh information bits
in the top $(m-n)$ levels. Using feedback, it decodes the $(m-n)$ bits
transmitted by the $(k+1)$th encoder. For any $1<t\leq K$,
the $k$th encoder transmits fresh information on the top-most $n$
levels and in the remaining $(m-2n)$ levels, it transmits the lower
$(m-2n)$ bits decoded at $(t-1)$.  This scheme achieves $Kn + (m-2n)$
bits per user in $K$ channel uses. Hence for $\alpha\geq 2$, we have
\begin{align}
\mathcal{C}^{\mathrm{FB}}_{\mathrm{sym},\mathrm{LD}}(\alpha,K)&\geq 1 + \frac{(\alpha-2)}{K}.
\end{align}
Theorem \ref{Theorem1} shows that the feedback coding scheme presented above is optimal.
Therefore, it is clear that the gain obtained via
feedback should decrease as the number of users increases. To substantiate this claim, we note that
when $\alpha\geq 2$ (corresponding to $m\geq 2n$), a normalized per-user rate of $1$ can always be achieved without feedback
by remaining silent on the lower most $(m-n)$ levels and sending fresh information in the top-most $n$ levels.
However, with feedback, each user can send additional information in the middle $(m-2n)$ levels in the first channel use. This additional
information can eventually reach the intended decoder via the delayed feedback path in $K$ channel uses. For instance, in Figure \ref{SchemeKeven3},
the bit $a_{2}$ is eventually received at decoder $1$ in the last channel use.
Therefore feedback can boost the reliable transmission from $Kn$ bits to $Kn+ (m-2n)$ bits in $K$ uses of the channel. Thus, the
rate gain obtained via feedback is $(m-2n)/K$ which decreases as $K$ increases.

\subsection{Coding for LD-CZIC with Global Feedback}\label{coding:globalFB}
We presented the effects of local feedback on the DoF in the previous section, which vanish as K grows for the cyclic network. However, we do not claim the same behavior for a general topology and feedback model. In this section we study the cyclic network under the deterministic model with global feedback; that is each transmitter receives the output signal of \emph{all} the receivers with a unit delay. We will show that the V-curve for the 2-user channel is still a valid characterization for the sum capacity of this network.   

We present the encoding scheme with global feedback for $K=3$.  To this end, we will show that the following rate-triple  is achievable with global feedback:
\begin{align}
(R_{1},R_{2},R_{3})&= \left(R^{*},R^{*},R^{*}\right),
\end{align} 
where
\begin{align}
R^{*}&= \frac{\max(2n-m,m)}{2}.
\end{align}

In the following subsections, we describe the scheme for specific values of $(n,m)$. The generalization to arbitrary $K$ and arbitrary $(n,m)$ is straightforward.

\subsubsection{Weak Interference ($n\geq m$) }
We illustrate the basic idea behind the scheme when $n=3$ and $m=1$ so that $\alpha=1/3$. Figure \ref{FullFBweak} shows a scheme that 
achieves a rate of $5/2$ per user via global feedback.  At $t=1$, all transmitters send fresh information on all $n=3$ levels.
The least significant $m=1$ bits suffer interference from the adjacent
transmitters. Upon receiving global feedback, all $Kn=3n=9$ bits can be
obtained at all the transmitters. This is the additional benefit of
global feedback. Subsequently at $t=2$, each transmitter sends the
bits $a_{3}, b_{3}$ and $c_{3}$ respectively on the top $m=1$ levels. In the remaining lower $(n-m)$ levels, transmitter $j$ sends fresh information bits.
However, in the lower-most $m=1$ level, transmitter $j$ cancels the interference that receiver $j$ will suffer from transmitter $j+1$ at $t=2$ (since it is known via global feedback). In particular transmitter $1$ sends $a_{5}\oplus b_{1}$ in the lowest level at $t=2$ etc. It is clear that $5$ bits are reliably transmitted to each receiver in $2$ channel uses, thus achieving the rate of $(2m-m)/2= 5/2$ bits/channel-use per user.

\begin{figure*}[t]
\centering
\subfigure{
\includegraphics[width=0.44\textwidth]{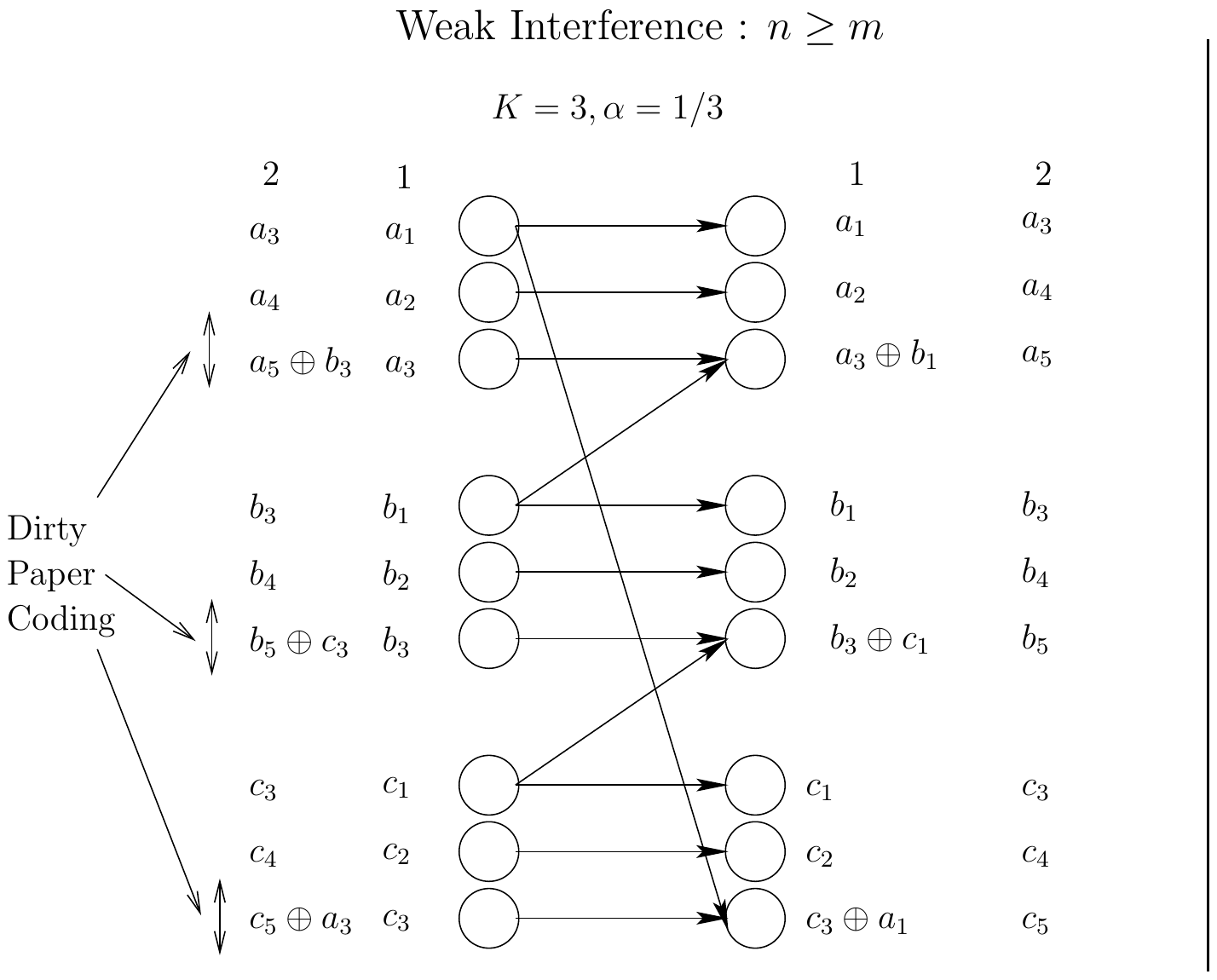}
}
\subfigure{
\includegraphics[width=0.44\textwidth]{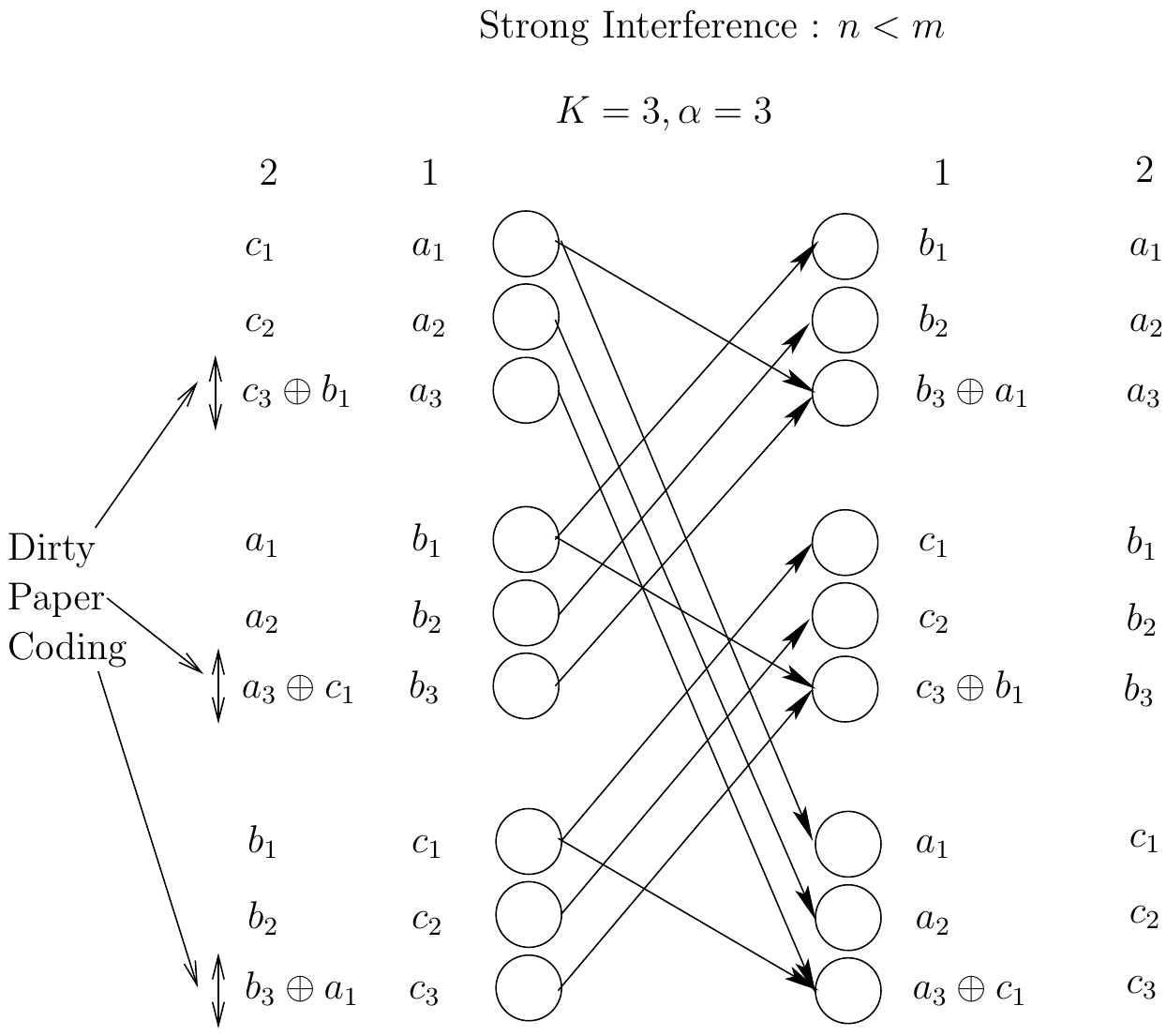}
}
  \caption{Coding for LD-CZIC with global feedback.}\label{FullFBweak}
  \vspace{-0.6cm}
\end{figure*}
\subsubsection{Strong Interference ($n< m$) }
For the strong interference case, we focus on $n=1$ and $m=3$ so that $\alpha=3$.
Figure \ref{FullFBweak} shows a scheme that achieves a rate of $3/2$ per user via global feedback.
At $t=1$, all transmitters send fresh information on all $m=3$ levels.
transmitters. Upon receiving global feedback, all $Km=3m=9$ bits can be
obtained at all the transmitters. This is the additional benefit of
global feedback. Subsequently at $t=2$, transmitter $j$ sends the bits required by receiver $(j-1)\mod K$.
In addition, in the lower-most $n=1$ level, it cancels the interference that receiver $(j-1)\mod K$ will face at $t=2$ from its own transmitter. 
It is clear that $m=3$ bits are reliably transmitted to
each receiver in $2$ channel uses, thus achieving the rate of $m/2= 3/2$ bits/channel-use per user.

From these schemes, it becomes clear that it is the local feedback constraint that causes the gain due to feedback to decrease
as $K$ increases. As we have shown, under the global feedback assumption, the idea of canceling known interference (commonly referred to as
dirty paper coding) can be employed to obtain the $V$-curve. 

\section{Upper bounds on the Feedback Sum-Capacity}\label{Section:Upper}
In this section, we present two types of upper bounds on the sum-capacity of the
$K$-user LD-CZIC. The type-I upper bound allows us to show that the normalized symmetric feedback capacity
for the $K$-user LD-CZIC is always upper bounded by the symmetric feedback capacity of the $2$-user system.
The type-II upper bound is in fact a set of $K!$ genie-aided upper bounds,
in which each upper bound corresponds to a permutation of $K$ users. These type-II upper bounds are in fact
valid for the general $K$-user interference channel with noiseless channel output feedback, i.e., they are not specifically derived
for the cyclic interference channel.

We present the type-I upper bound in the following theorem:
\begin{Theo}\label{theoremK2}
The normalized symmetric feedback capacity of the $K$-user LD-CZIC satisfies
\begin{align}
\mathcal{C}^{\mathrm{FB}}_{\mathrm{sym},\mathrm{LD}}(\alpha,K)&\leq \max\left(1-\frac{\alpha}{2},\frac{\alpha}{2}\right).
\end{align}
\end{Theo}
The proof of Theorem \ref{theoremK2} is given in the appendix. The main idea behind this upper bound is to show that
\begin{align}
R_{j}+R_{(j+1)}\leq \max(2n-m,m),
\end{align}
for $j=1,\ldots,K$. By adding all such $K$ upper bounds and normalizing by $2nK$, we obtain the
desired bound stated in Theorem  \ref{theoremK2}.
Theorem \ref{theoremK2} along with (\ref{nofeedback}) leads to the conclusion that feedback does
not increase the symmetric capacity of the $K$-user LD-CZIC in the regime $\alpha\in [2/3,2]$. We also remark here 
that this upper bound also holds under the global feedback assumption. The intuition behind this can be seen as follows: the pairwise upper bounds
can be regarded as genie aided bounds in which the messages of the remaining $(K-2)$ users are supplied to both receivers and both transmitters,
which is tantamount to global feedback.

We next present the type-II upper bound:
\begin{Theo}\label{theoremSUMUB}
Fix a permutation order $\pi=\{\pi_{1},\ldots,\pi_{K}\}$. Then the feedback
sum-capacity of the general $K$-user interference channel is upper bounded as follows:
\begin{align}
&\mathcal{C}_{\mathrm{sum}}^{\mathrm{FB}}(K)\nonumber\\
&\leq \max_{p(x_{1},\ldots,x_{K})}\sum_{k=1}^{K}\bigg[H(Y_{\pi_{k}}|X_{\pi_{1}},Y_{\pi_{1}},\ldots,X_{\pi_{k-1}},Y_{\pi_{k-1}}) \nonumber\\
&\hspace{3cm}- H(Y_{1},\ldots,Y_{K}|X_{1},\ldots,X_{K})\bigg].\nonumber
\end{align}
\end{Theo}
To illustrate  by an example, consider the case in which $K=3$, for which Theorem \ref{theoremSUMUB} yields $6$ upper bounds on the feedback sum capacity:
\begin{align}
&\max_{p(x_{1},x_{2},x_{3})}\Big[H(Y_{1})+ H(Y_{2}|X_{1},Y_{1}) +
H(Y_{3}|X_{1},X_{2},Y_{1},Y_{2})\nonumber\\&\hspace{2cm}- H(Y_{1},Y_{2},Y_{3}|X_{1},X_{2},X_{3})\Big]\nonumber\\
&\max_{p(x_{1},x_{2},x_{3})}\Big[H(Y_{1})+ H(Y_{3}|X_{1},Y_{1}) +
H(Y_{2}|X_{1},X_{3},Y_{1},Y_{3})\nonumber\\&\hspace{2cm}- H(Y_{1},Y_{2},Y_{3}|X_{1},X_{2},X_{3})\Big]\nonumber\\
&\max_{p(x_{1},x_{2},x_{3})}\Big[H(Y_{2})+ H(Y_{1}|X_{2},Y_{2}) +
H(Y_{3}|X_{1},X_{2},Y_{1},Y_{2})\nonumber\\&\hspace{2cm}- H(Y_{1},Y_{2},Y_{3}|X_{1},X_{2},X_{3})\Big]\nonumber\\
&\max_{p(x_{1},x_{2},x_{3})}\Big[H(Y_{2})+ H(Y_{3}|X_{2},Y_{2}) +
H(Y_{1}|X_{2},X_{3},Y_{2},Y_{3})\nonumber\\&\hspace{2cm}- H(Y_{1},Y_{2},Y_{3}|X_{1},X_{2},X_{3})\Big]\nonumber
\end{align}
\begin{align}
&\max_{p(x_{1},x_{2},x_{3})}\Big[H(Y_{3})+ H(Y_{1}|X_{3},Y_{3}) +
H(Y_{2}|X_{1},X_{3},Y_{1},Y_{3})\nonumber\\&\hspace{2cm}- H(Y_{1},Y_{2},Y_{3}|X_{1},X_{2},X_{3})\Big]\nonumber\\
&\max_{p(x_{1},x_{2},x_{3})}\Big[H(Y_{3})+ H(Y_{2}|X_{3},Y_{3}) +
H(Y_{1}|X_{2},X_{3},Y_{2},Y_{3})\nonumber\\&\hspace{2cm}- H(Y_{1},Y_{2},Y_{3}|X_{1},X_{2},X_{3})\Big]\nonumber.
\end{align}
For an arbitrary $K$, Theorem \ref{theoremSUMUB} gives a total of $K!$
upper bounds. Optimization of these bounds for the general $K$ user
case and asymmetric channel gains is prohibitively
complex. For the scope of this paper, we are interested in the case of CZIC with symmetric channel parameters.
Depending on the range of the interference parameter $\alpha$, we carefully select one of the type-II
bounds and evaluate it to obtain the desired converse result as stated in Theorem \ref{Theorem1}.

\subsection{Very Weak and Weak interference regimes: $0\leq \alpha \leq 2/3$}
In this regime, we select the type-II upper bound corresponding to the identical permutation order:
\begin{align}
\pi&= (1,2,\ldots,K).
\end{align}
Theorem \ref{theoremSUMUB} yields the following  bound on the
sum-capacity:
\begin{align}
&\mathcal{C}_{\mathrm{sum},\mathrm{LD}}^{\mathrm{FB}}(K)\nonumber\\
&\leq
\max_{p(x_{1},\ldots,x_{K})}\bigg[\sum_{k=1}^{K}H(Y_{k}|X_{1},Y_{1},\ldots,X_{k-1},Y_{k-1}) \nonumber\\&\hspace{2cm}- H(Y_{1},\ldots,Y_{K}|X_{1},\ldots,X_{K})\bigg]\\
&=
\max_{p(x_{1},\ldots,x_{K})}\sum_{k=1}^{K}H(Y_{k}|X_{1},Y_{1},\ldots,X_{k-1},Y_{k-1})\label{eqna2}\\
&= \max_{p(x_{1},\ldots,x_{K})} \bigg[ H(Y_{1}) + H(Y_{2}|X_{1},Y_{1}) +\ldots\nonumber\\&\hspace{2cm}+
H(Y_{K}|X_{1},Y_{1},\ldots,X_{K-1},Y_{K-1})\bigg]\label{eqna3}\\
&\leq n + \max_{p(x_{1},\ldots,x_{K})}\sum_{k=2}^{K-1}H(Y_{k}|X_{k-1},Y_{k-1}) \nonumber\\&\hspace{2cm}+\max_{p(x_{1},\ldots,x_{K})}H(Y_{K}|X_{1},Y_{1},X_{K-1},Y_{K-1})\label{eqna4},
\end{align}
where (\ref{eqna2}) follows from the fact that the channel outputs $(Y_{1},\ldots,Y_{K})$ are deterministic functions of the channel inputs $(X_{1},\ldots,X_{K})$,
and (\ref{eqna4}) follows from the fact that $H(Y_{1})\leq \max(m,n)=n$.

To further upper bound (\ref{eqna4}), we first recall the notation used for $n\geq m$ in (\ref{notationweak}):
\begin{align}
U_{k}&: \mbox{ top-most } (n-m) \mbox{ bits of } X_{k}\nonumber\\
V_{k}&: \mbox{ top-most } m \mbox{ bits of } X_{k}\nonumber\\
L_{k}&: \mbox{ lower-most } m \mbox{ bits of } X_{k}\nonumber.
\end{align}
For any $2\leq k\leq (K-1)$, we have the following sequence of inequalities:
\begin{align}
H(Y_{k}|X_{k-1},Y_{k-1})
&= H(Y_{k}|X_{k-1},Y_{k-1},V_{k})\label{eqnb1}\\
&= H(U_{k},L_{k}\oplus V_{(k+1)}|X_{k-1},Y_{k-1},V_{k})
\end{align}
\begin{align}
&\leq H(U_{k}|V_{k})+ H(L_{k}\oplus V_{k+1}) \\
&\leq \max\left(0,n-2m\right)+m\label{midterm},
\end{align}
where (\ref{eqnb1}) is due to the fact that $V_{k}$ can be determined from $(X_{k-1},Y_{k-1})$.

Finally we upper bound the last term in (\ref{eqna4}) as follows:
\begin{align}
&H(Y_{K}|X_{1},Y_{1},X_{K-1},Y_{K-1})\nonumber\\
&= H(Y_{K}|V_{1},X_{1},Y_{1},X_{K-1},Y_{K-1},V_{K})\\
&= H(U_{K},L_{K}\oplus V_{1}|V_{1},V_{K},X_{1},Y_{1},X_{K-1},Y_{K-1})\\
&\leq H(U_{K},L_{K}|V_{K})\\
&= H(X_{K}|V_{K})\\
&\leq (n-m).\label{lastterm}
\end{align}

Using (\ref{midterm}) and (\ref{lastterm}), we can further upper bound (\ref{eqna4}) to obtain
\begin{align}
&\mathcal{C}_{\mathrm{sum},\mathrm{LD}}^{\mathrm{FB}}(K)\nonumber\\
&\leq n +  (K-2)\Big[\max\left(0,n-2m\right)+m\Big] + (n-m).
\end{align}
Therefore, the normalized symmetric feedback capacity is upper bounded as follows:
\begin{align}
\mathcal{C}^{\mathrm{FB}}_{\mathrm{sym},\mathrm{LD}}(\alpha,K)&\leq \max(\alpha,1-\alpha) + \frac{\min(\alpha,2-3\alpha)}{K},
\end{align}
which can also be written as
\begin{align}
\mathcal{C}^{\mathrm{FB}}_{\mathrm{sym},\mathrm{LD}}(\alpha,K)\leq
\begin{cases}
(1-\alpha)+\frac{\alpha}{K}, &0\leq\alpha\leq 1/2\\
\alpha+ \frac{(2-3\alpha)}{K}, &1/2\leq \alpha\leq 2/3.
\end{cases}
\end{align}

Note that the upper bound alone shows that in the limit $K\rightarrow \infty$ the
upper bound converges to the no-feedback symmetric capacity.
This implies that in the limit of large $K$, the feedback gain vanishes.

\subsection{Very strong interference: $\alpha \geq 2$}
In this regime, we select the type-II upper bound corresponding to the following permutation order:
\begin{align}
\pi&= (1,K,K-1,K-2,\ldots,3,2).
\end{align}
Theorem \ref{theoremSUMUB} yields the following upper bound on the
sum-capacity:
\begin{align}
&C_{\mathrm{sum},\mathrm{LD}}^{\mathrm{FB}}(K)\nonumber\\
&\leq
\max_{p(x_{1},\ldots,x_{K})}\Big[ H(Y_{1}) +
 H(Y_{K}|X_{1},Y_{1})+\ldots\nonumber\\&\hspace{2.1cm}\ldots+H(Y_{2}|X_{1},X_{3},\ldots,X_{K},Y_{1},Y_{3},\ldots,Y_{K})\Big]\\
&\leq \max_{p(x_{1},\ldots,x_{K})}\Big[ H(Y_{1}) + H(Y_{K}|X_{1},Y_{1})\nonumber\\
&\hspace{1.2cm}  + \sum_{k=3}^{K-1}H(Y_{k}|X_{k+1},Y_{k+1})+ H(Y_{2}|X_{1},Y_{1},X_{3},Y_{3})\Big]\label{eqnSa}.
\end{align}
To further upper bound (\ref{eqnSa}), we recall the notation used for $n< m$  in (\ref{notationstrong}):
\begin{align}
U_{k}&: \mbox{ top-most } (m-n) \mbox{ bits of } X_{k}\nonumber\\
V_{k}&: \mbox{ top-most } n \mbox{ bits of } X_{k}\nonumber\\
L_{k}&: \mbox{ lower-most } n \mbox{ bits of } X_{k}\nonumber.
\end{align}
We now upper bound the terms in (\ref{eqnSa}) as follows. We first have the trivial upper bound $H(Y_{1})\leq \max(m,n)=m$.
We then bound the second term in (\ref{eqnSa}) as follows:
\begin{align}
H(Y_{K}|X_{1},Y_{1})&= H(U_{1},L_{1}\oplus V_{K}|X_{1},Y_{1})\\
&= H(L_{1}\oplus V_{K}|X_{1},Y_{1})\\
&\leq n\label{eqnSb}.
\end{align}
Next, for any $3\leq k\leq (K-1)$, we have
\begin{align}
H(Y_{k}|X_{k+1},Y_{k+1})&= H(U_{k+1},L_{k+1}\oplus V_{k}|X_{k+1},Y_{k+1})\\
&= H(L_{k+1}\oplus V_{k}|X_{k+1},Y_{k+1})\\
&\leq n,
\end{align}
which implies that
\begin{align}
\sum_{k=3}^{K-1}H(Y_{k}|X_{k+1},Y_{k+1})&\leq (K-3)n\label{eqnSc}.
\end{align}
Finally, we have
\begin{align}
H(Y_{2}|X_{1},Y_{1},X_{3},Y_{3})
&= H(U_{3},L_{3}\oplus V_{2}|X_{1},Y_{1},X_{3},Y_{3})\\
&= H(V_{2}|X_{1},Y_{1},X_{3},Y_{3})\\
&= H(V_{2}|U_{2},X_{1},Y_{2},X_{3},Y_{3})\\
&= 0\label{eqnSd},
\end{align}
where (\ref{eqnSd}) follows from the fact that $\alpha\geq 2$ corresponds to the case in which $m-n\geq n$
and therefore $V_{2}$ is completely determined by $U_{2}$.

Using (\ref{eqnSb}), (\ref{eqnSc}) and (\ref{eqnSd}), we have the following upper bound from (\ref{eqnSa}):
\begin{align}
&\mathcal{C}_{\mathrm{sum},\mathrm{LD}}^{\mathrm{FB}}(K)\nonumber\\
&\leq H(Y_{1}) + H(Y_{K}|X_{1},Y_{1}) + \sum_{k=3}^{K-1}H(Y_{k}|X_{k+1},Y_{k+1}) \nonumber\\&\hspace{0.5cm}+ H(Y_{2}|X_{1},Y_{1},X_{3},Y_{3})\\
&\leq m + (K-2)n.
\end{align}
Normalizing this upper bound by $nK$, we obtain
\begin{align}
\mathcal{C}^{\mathrm{FB}}_{\mathrm{sym},\mathrm{LD}}(\alpha,K)&= \frac{\mathcal{C}_{\mathrm{sum},\mathrm{LD}}^{\mathrm{FB}}(K)}{nK}\\
&\leq \frac{m+(K-2)n}{nK}\\
&= 1 + \frac{(\alpha-2)}{K},
\end{align}
which is the desired upper bound on the normalized symmetric feedback capacity.

\section{Gaussian $K$-user CZIC with Feedback}\label{GaussianSection}
In this section, we consider the $K$-user Gaussian CZIC with feedback.
The signal transmitted by user $k$ is denoted by $X_k$.
We impose an average unit power constraint at each user; that is $\E[X_k^2]\leq 1$. The signal observed at receiver $k$ is obtained by
\begin{align}
Y_k= \sqrt{\SNR} X_k+ \sqrt{\INR} X_{k+1} +Z_k, \qquad k=1,2,\dots,K,
\label{eq:G-model}
\end{align}
where we define $X_{K+1}=X_1$ for consistency.

In the following, we study five different regimes depending on the parameter $\alpha$ (again,
defined for this model as $\alpha= \log(\INR)/\log(\SNR)$, and propose upper bounds and feedback coding schemes for each one.
We analyze the performance of the proposed schemes, and derive symmetric achievable rates for them. In the rest of this section, we use bold symbols to denote blocks of length $T$, e.g.,
\[
\bx_k[j]\hspace{-2pt}
=\hspace{-2pt}\left (X_k((j-1)T\hspace{-2pt}+\hspace{-2pt}1),\hspace{-1pt} (X_k((j-1)T\hspace{-1.2pt}+\hspace{-1.2pt}2), \dots, (X_k(jT) \right).
\]
The encoding schemes that we propose for each regime involve message splitting and power/rate allocation to the resulting sub-messages, motivated by the analysis of the linear deterministic model. The message splitting at the encoders is similar to that we have seen for the LD model. The powers allocated to the sub-messages at the $k$th transmitter are chosen so that they are received at the $k$th and $(k-1)$th receiver at a proper power. For sake of clarity, we explain the relationship between the coding scheme for the LD model and the Gaussian model for one of the regimes (the very weak interference regime) in full detail. However, we avoid repeating the same argument for other regimes as they follow similarly.

Moreover, for each regime of parameters the power/rate allocations proposed in the following subsections are only meaningful under a certain underlying assumption on the values of $\SNR$, $\INR$, and their proportion. More precisely, in our analysis we have excluded some marginal ranges of $\SNR$ and $\INR$ for which the rates are constant. It is worth mentioning that the $K$-user cyclic Z-Interference channel studied in this work can be approximated by simple and easily analyzable models for the excluded range of parameters. We present analysis of the excluded ranges for the very weak interference regime in Appendix~\ref{app: excluded-range} for completeness, and avoid repeating a similar
argument for other regimes for the sake of brevity.

We state our upper bounds in terms of following expressions:
\begin{align}
A&\triangleq \frac{1}{2}\log\left(1+\SNR+\INR+2\sqrt{\SNR\cdot\INR}\right)\\
B&\triangleq \frac{1}{2}\log\left(1+\SNR+2\INR+\INR^{2}+2\sqrt{\SNR\cdot\INR}\right)\\
C&\triangleq \frac{1}{2}\log\left(1+\SNR+\INR\right)\\
D&\triangleq \frac{1}{2}\log(1+\SNR)\\
E&\triangleq \frac{1}{2}\log(1+\INR).
\end{align}

\subsection{Very Weak Interference: $0\leq \alpha\leq 1/2$}
\label{seubsec:G-VW}
\subsubsection{Coding Scheme}
Recall the coding strategy proposed for the LD model under the very weak interference regime in Section~\ref{sec:LD-VW}, where a total of $K(n-m)+m$ information symbols were transmitted to each receiver over $K$ channel uses. There each transmit signal is split into three set of mutually exclusive levels, namely $X_{k}(t)=(X_{k,1}(t),X_{k,2}(t),X_{k,3}(t))$, and user $k$ sends its information symbols in $X_{k,1}(1)$ (including $m$ bits), $\{X_{k,2}(1),\dots,X_{k,2}(K)\}$ (including $(n-2m)K$ bits), and $\{X_{k,3}(1),\dots,X_{k,3}(K)\}$ (including $mK$ bits).
Moreover, $X_{k,1}(t)$ includes those levels that are heard at the $(k-1)$th receiver as interference, and the received power of $(X_{k,2},X_{k,3}(t))$ at the $(k-1)$th receiver is at the noise level. 

Motivated by the summary above, the encoding scheme we propose here also takes $K$ blocks, each of length $T$.
We split the message of user $k$ into a total of $(2K+1)$ messages, namely \[
\left( M_k^{(h)},M_k^{(1,m)},M_k^{(2,m)}\dots, M_k^{(K,m)}, M_k^{(1,l)}, M_k^{(2,l)}, M_k^{(K,l)} \right),
 \]
 that it wishes to send to its respective receiver over $K$ transmission blocks. $M_k^{(h)}$, $M_k^{(t,m)}$ and $M_k^{(t,l)}$ in this message splitting are respectively the counterparts of $X_{k,1}(1)$, $X_{k,2}(t)$ and $X_{k,3}(t)$.\\
 Moreover, we set the size of these message sets with rates given by
\begin{align}
\begin{split}
\frac{\log |\cM_k^{(h)}|}{n}&= R_1=\cgfp{\INRP+1}{3}, \\&\qquad k=1,\dots,K,  \\
\frac{\log |\cM_k^{(j,m)}|}{n}&= R_2=\cgfp{\frac{\SNRP}{\INRP}+1}{2\INRP+1}, \\&\qquad k=1,\dots,K,\ j=1,\dots,K,\\
\frac{\log |\cM_k^{(j,l)}|}{n}&= R_3= \cgfp{\INRP+1}{2}, \\&\qquad k=1,\dots,K, \ j=1,\dots,K,
\end{split}
\label{eq:rates-vw}
\end{align}
where $\log^+(x)=\max(\log(x),0)$.
The messages are encoded using  individual Gaussian codebooks with unit average power, to obtain
\[
\left(\bs_k^{(h)},\bs_k^{(1,m)},\dots, \bs_k^{(K,m)},  \bs_k^{(1,l)}, \bs_k^{(K,l)} \right).
\]
The signals transmitted by user $k$ in block $k$ are formed as combinations of such Gaussian codewords through a proper message splitting, i.e.,
\[
\bx_k[j]= \beta_h \bx_{k,h}[j] + \beta_m \bx_{k,m}[j] + \beta_l \bx_{k,l}[j] ,
\]
where the power factors $\beta_h$, $\beta_m$ and $\beta_l$ can be chosen such that $\beta_h^2+\beta_m^2+\beta_l^2\leq 1$.
Here, we have $\bx_{k,m}[j]=\bs_{k}^{(j,m)}$ and $\bx_{k,l}=\bs_k^{(j,l)}$. Moreover, in the first block, we set $\bx_{k,h}[1]=\bs_{k}^{(h)}$.
For subsequent blocks, the high power part of the signal consists of the high power codeword of the neighboring transmitter sent over the last block, that is
$\bx_{k,h}[j]=\bx_{k+1,h}[j-1]$. Of course, this is only possible if transmitter $k$ can decode $\bx_{k+1,h}[j-1]$ from $\by_k[j-1]$ received over the feedback link.

It remains to set the fraction of power allocated to each part of the transmit signal. In the rest of this section we assume that $\INR\geq 2$, and $\SNR \geq 2\INR^2$, which are conditions that ensure that all the rates in \eqref{eq:rates-vw} are positive. We will separately analyze the excluded range of $\SNR$ and $\INR$ in Appendix~\ref{app: excluded-range}. Having the fact that $(\bx_{k,m}[j], \bx_{k,l}[j] )$ are the counterparts of $(X_{k,2}(j),X_{k,3}(j))$ in the LD model, we allocate them a fraction of power so that they are received at the $(k-1)$th receiver at the noise level. This allows us to safely treat them as noise when decoding at the $(k-1)$th receiver. Note from \eqref{eq:G-model} that the power of $(\bx_{k,m}[j], \bx_{k,l}[j] )$ would be magnified by $\INR$ at the $(k-1)$th receiver, and we wish the result to be at the noise power. Hence we allocate a total of $1/\INR$ of the available power at the $k$th transmitter to these codewords, and the remaining $(\INR-1)/\INR$ to the high power part, i.e.,
\[
\beta_h=\sqrt{\frac{\INR-1}{\INR}}.
\]
On the other hand, $\bx_{k,l}[j]$ is the counterpart of $X_{k,3}(j)$ which is corrupted by interference when received at its intended receiver. Remember from  Section~\ref{sec:LD-VW} that we can only decode $X_{k,3}(j)$ once the interfering signal is decoded and removed in the next channel use. The power of the interference at the $k$th receiver is $1/\INR$. Hence, we choose the power of $\bx_{k,l}[j]$ so that after magnification by $\SNR$ over the channel to the $k$th receiver, it gets the same power as the interference, i.e., we choose
\[
\beta_l=\sqrt{\frac{\INR}{\SNR}},
\]
and the remaining power will be allocated to $\bx_{k,m}[j]$, that is
\[
\beta_m=\sqrt{\frac{\SNR-\INR^2}{\SNR\cdot\INR}}.
\]
Therefore, the transmit signal from the $k$th transmitter over the $j$th block can be written as
\begin{align}
\bx_k[j]&= \sqrt{\frac{\INRP-1}{\INRP}} \bx_{k,h}[j] + \sqrt{\frac{\SNRP-\INRP^2}{\SNRP \cdot \INRP}}\bx_{k,m}[j] \nonumber\\&\quad+ \sqrt{\frac{\INRP}{\SNRP}}\bx_{k,l}[j].\nonumber
\end{align}

\paragraph{Decoding the feedback signal at encoder}
Upon receiving $\by_k[j-1]$, transmitter $k$ removes its own signal to obtain
\begin{align}
&\by_k[j-1]-\sqrt{\SNRP} \bx_k[j-1] \nonumber\\
&= \sqrt{\INRP} \bx_{k+1}[j-1] +\bz_k[j-1]\nonumber\\
&=\sqrt{\INRP-1} \bx_{k+1,h}[j-1] +\sqrt{\frac{\SNRP-\INRP^2}{\SNRP}} \bx_{k+1,m}[j-1] \nonumber\\&\qquad+ \sqrt{\frac{\INRP^2}{\SNRP}} \bx_{k+1,l}[j-1] +\bz_k[j-1].
\end{align}
It then decodes $\bx_{k+1,h}[j-1]$, up to rate
\begin{align}
R_1 \leq \cgf{\INRP+1}{2}.\label{const:VW-1}
\end{align}

\paragraph{Decoding process at the receiver}
The receiver node $k$ has to decode $\bx_{k,h}[j]$ and $\bx_{k,m}[j]$. This is done using a sequential decode-and-remove scheme, which can support rates that satisfy
\begin{align}
R_1 \leq \cgf{\SNRP+\INRP+1}{\frac{\SNRP}{\INRP}+\INRP+1}\label{const:VW-2}\\
R_2 \leq \cgf{\frac{\SNRP}{\INRP}+\INRP+1}{2\INRP +1}.\label{const:VW-3}
\end{align}
The receiver also stores the remaining part of its received signal,
\begin{align}
&\tilde{\by}_k[j] \nonumber\\
&= \sqrt{\INRP} \bx_{k,l}[j] + \sqrt{\INRP-1} \bx_{k+1,h}[j] \nonumber\\
&\quad+ \sqrt{\frac{\SNRP-\INRP^2}{\SNRP}} \bx_{k+1,m}[j] + \sqrt{\frac{\INRP^2}{\SNRP}} \bx_{k+1,l}[j] + \bz_k[j],
\end{align}
for further processing. In the next block, upon decoding $\bx_{k,h}[j+1]$, it can use it to remove a part of the interference in $\tilde{\by}_k[j]$. Recall that $\bx_{k,h}[j+1]=\bx_{k+1,h}[j]$. Therefore, by removing $\bx_{k+1,h}[j]$ from $\tilde{\by}_k[j]$, it can decode $\bx_{k,l}[j]$, provided that its rate satisfies
\begin{align}
R_3 \leq \cgf{\INRP+2}{2}.\label{const:VW-4}
\end{align}
The total achievable rate would be
\begin{align}
R_{\mathrm{sym}}&=\frac{ R_1 + K R_2 + K R_3}{K}\\
&\geq \cg{\frac{\SNRP}{\INRP}} +\frac{1}{2K} \log \left(1+\INRP\right) \nonumber\\
&\qquad-\frac{1}{2}\log\left(\frac{2(2\INR+1)}{\INR+1}\right)-\frac{\log 3}{2K}\\
&\geq (D-E)+ \frac{E}{K}-1-\frac{\log 3}{2K}.
\end{align}

\subsubsection{Upper Bound}
In this regime, we use the following upper bound from Theorem \ref{theoremSUMUB} on the feedback sum-capacity:
\begin{align}
&\mathcal{C}_{\mathrm{sum},\mathrm{G}}^{\mathrm{FB}}(K)\nonumber\\
&\leq \max_{p(x_{1},\ldots,x_{K})} \Big[h(Y_{1})+h(Y_{2}|X_{1},Y_{1})+\ldots\nonumber\\
&\hspace{2.1cm}+h(Y_{K}|X_{1},Y_{1},\ldots,X_{K-1},Y_{K-1})\nonumber\\
&\hspace{2.1cm}- h(Y_{1},\ldots,Y_{K}|X_{1},\ldots,X_{K})\Big].\label{WeakGaussian}
\end{align}
We first note the following:
\begin{align}
h(Y_{1})&\leq \frac{1}{2}\log\left(1+\SNR+\INR+2\sqrt{\SNR\cdot\INR}\right) + c\\
&=  A + c,
\end{align}
where $c=1/2\log(2\pi\mbox{e})$.

For $2\leq k\leq (K-1)$, we have
\begin{align}
&h(Y_{k}|X_{1},Y_{1},\ldots,X_{k-1},Y_{k-1})\nonumber\\
&\leq h(Y_{k}|X_{k-1},Y_{k-1})\\
&= h(\sqrt{\SNR}X_{k}+\sqrt{\INR}X_{k+1}+Z_{k}|\sqrt{\INR}X_{k}+Z_{k-1},X_{k-1})\\
&\leq h(\sqrt{\SNR}X_{k}+\sqrt{\INR}X_{k+1}+Z_{k}|\sqrt{\INR}X_{k}+Z_{k-1})\\
&\leq  \frac{1}{2}\log\left(\frac{1+\SNR+2\INR+\INR^{2}+2\sqrt{\SNR\cdot\INR}}{1+\INR}\right)  + c\\
&= B-E+c.
\end{align}
Similarly, we have
\begin{align}
&h(Y_{K}|X_{1},Y_{1},\ldots,X_{K-1},Y_{K-1})\nonumber\\
&\leq h(Y_{K}|X_{1},Y_{1},X_{K-1},Y_{K-1})\\
&\leq h(\sqrt{\SNR}X_{K}+Z_{K}|\sqrt{\INR}X_{K}+Z_{K-1})\\
&\leq \frac{1}{2}\log\left(\frac{1+\SNR+\INR}{1+\INR}\right)+c\\
&= C-E+c.
\end{align}
Finally, we have
\begin{align}
h(Y_{1},\ldots,Y_{K}|X_{1},\ldots,X_{K})&= h(Z_{1},\ldots,Z_{K})\\
&= Kc.
\end{align}
Hence, from (\ref{WeakGaussian}), we have
\begin{align}
\mathcal{C}_{\mathrm{sum},\mathrm{G}}^{\mathrm{FB}}(K)&\leq A + (K-2)(B-E) + C-E\\
&= K(B-E)+(A+C+E-2B),\label{Twousages}
\end{align}
which implies that
\begin{align}
\mathcal{C}_{\mathrm{sym},\mathrm{G}}^{\mathrm{FB}}(K)&\leq (B-E) +\frac{A+C+E-2B}{K}\\
&= (B-E) + \frac{E}{K} + \frac{(A+C-2B)}{K}\\
&\leq (B-E) + \frac{E}{K},
\label{eq:UB:G-VW}
\end{align}
where we have used the fact that
\begin{align}
A&\leq B\\
C&\leq B,
\end{align}
which implies that $(A+C-2B)\leq 0$.

We also note that
\begin{align}
2B&= \log(1+ \SNR+2\INR+ \INR^{2}+2\sqrt{\SNR\cdot\INR})\\
&\leq \log(1+ 6\SNR)\\
&\leq \log(6)+ \log(1+\SNR)\\
&= \log(6)+2D.
\end{align}
Collecting all the bounds, we have
\begin{align}
\mathcal{C}_{\mathrm{sym},\mathrm{G}}^{\mathrm{FB}}(K)&\leq(B-E) + \frac{E}{K} \\
&\leq (D-E)+\frac{E}{K} + \frac{\log(6)}{2}.
\end{align}
Hence, the symmetric feedback capacity satisfies
\begin{align}
&\Bigg[(D-E)+ \frac{E}{K}\Bigg]-1-\frac{\log 3}{2K}\nonumber\\
&\hspace{2cm} \leq \mathcal{C}_{\mathrm{sym},\mathrm{G}}^{\mathrm{FB}}(K)\nonumber\\
&\hspace{2cm}\leq \Bigg[(D-E)+\frac{E}{K}\Bigg] + \frac{\log(6)}{2}
\end{align}
so that the gap is given as
\begin{align}
\Delta&= \frac{\log(6)}{2}+1+ \frac{  \log(3)}{2K}\\
&< \frac{3}{2}+1+ \frac{2}{2K}\\
&\leq \frac{11}{4},
\end{align}
and the degrees of freedom are given as
\begin{align}
\DOF^{\mathrm{FB}}(\alpha,K)&= (1-\alpha)+\frac{\alpha}{K}, \quad 0\leq \alpha\leq 1/2.
\end{align}

\subsection{Weak Interference: $1/2\leq \alpha\leq 2/3$}
\subsubsection{Coding Scheme}
In this regime we have $\SNRP^{1/2} \leq \INRP \leq \SNRP^{2/3}$. The encoding scheme for this regime takes advantage of the feedback link. We create a cycle of length $K$ consisting of all the interfering and feedback links. A part of the message of each user is conveyed  through this cycle.

The encoding scheme is performed over $K$ blocks. Assume each user has $2K+1$ messages, namely
\[
\left( M_k^{(1,h)},M_k^{(2,h)},\dots, M_k^{(K,h)}, M_k^{(m)}, M_k^{(1,l)}, M_k^{(2,l)}, M_k^{(K,l)} \right),
 \]
 The rates of the messages are given by\footnote{Here we assume $\SNR\geq \INR \geq 1$ to guarantee positive rates. If this does not hold, similarly to the analysis in Appendix~\ref{app: excluded-range}, the bounded gap result can be established.}
\begin{align}
\begin{split}
\frac{\log |\cM_k^{(j,h)}|}{n}&= R_1=\cgfp{\INRP+1}{2\frac{\SNRP}{\INRP}+1}, \\&\qquad k=1,\dots,K, \ j=1,\dots,K,\\
\frac{\log |\cM_k^{(m)}|}{n}&= R_2= \frac{1}{2}\log^+\left(\frac{1+\SNR^{2}}{1+2\INR^{3}}\right),\\&\quad k=1,\dots,K,\\
\frac{\log |\cM_k^{(j,l)}|}{n}&= R_3=\cgfp{\frac{\SNRP}{\INRP} + 2}{2}, \\&\qquad k=1,\dots,K, \ j=1,\dots,K.
\end{split}
\label{eq:rates-w}
\end{align}

The $k$th transmitter encodes its message using a individual Gaussian codebook with unit average power, to obtain

\[
\left(\bs_k^{(1,h)},\dots, \bs_k^{(K,h)}, \bs_k^{(m)},  \bs_k^{(1,l)}, \bs_k^{(K,l)} \right).
\]
The transmitting signal in block $k$ is formed as
\begin{align}
\bx_k[j]&= \sqrt{\frac{\INRP^2-\SNRP}{\INRP^2}} \bx_{k,h}[j] + \sqrt{\frac{\SNRP-\INRP}{\INRP^2}}\bx_{k,m}[j]\nonumber\\&\qquad + \sqrt{\frac{1}{\INRP}}\bx_{k,l}[j] ,
\end{align}
where $\bx_{k,h}[j]=\bs_k^{(j,h)}$ and $\bx_{k,l}[j]=\bs_k^{(j,l)}$. The moderate power codeword transmitted during the first block is the codeword corresponding to $M_k^{(m)}$. However, in the next block, $\bx_{k,m}$ would be the moderate power codeword of the neighbor sent during the past block. More precisely,
\begin{align}
\bx_{k,m}[j]&=\bx_{k+1,m}[j-1]= \bx_{k+2,m}[j-2] \cdots = \bx_{k+j-1,m}[1]\nonumber\\&=\bs_{k+j-1}^{(m)}.\nonumber
\end{align}
Note that $\bx_{k+1,m}[j-1]$ has to be decoded from the signal sent to the transmitter $k$ over the feedback link at the end of block $(j-1)$.

\paragraph{Decoding the feedback signal at encoder}
The signal received at receiver $k$ in block $j$ is forwarded to its respective transmitter at the end of the block. The transmitter will use it for forming its transmitting signal in the next block.  The transmitter $k$ removes the part of the signal sent by it to obtain
\begin{align}
&\by_k[j]-\sqrt{\SNRP} \bx_k[j]\nonumber\\
&=\sqrt{\INRP}\bx_{k+1}[j] + \bz_{k}[j]\nonumber\\
&= \sqrt{\frac{\INRP^2-\SNRP}{\INRP}} \bx_{k+1,h}[j] + \sqrt{\frac{\SNRP-\INRP}{\INRP}}\bx_{k+1,m}[j] \nonumber\\&\qquad + \bx_{k+1,l}[j] + \bz_{k}[j].
\end{align}
The transmitter needs the moderate power codeword for the next transmission. In order to decode $\bx_{k+1,m}[j]$, it first decodes and removes the high power codeword, and then decodes the moderate power one.  This can be done provided
 \begin{align}
R_1 &\leq  \cgf{\INRP+1}{\frac{\SNRP}{\INRP}+1}\\
R_2 &\leq  \cgf{\frac{\SNRP}{\INRP}+1}{2}.
\end{align}
It is easy to check that the rates in \eqref{eq:rates-w} satisfy both constraints.

\paragraph{Decoding process at the receiver}
The decoding procedure at decoder $k$ is as follows. Upon receiving
\begin{align}
\by_k[j]&=\sqrt{\frac{\SNRP}{\INRP^2} (\INRP^2-\SNRP)} \bx_{k,h}[j] \nonumber\\
&\quad+ \sqrt{\frac{\SNRP}{\INRP^2}(\SNRP-\INRP)}\bx_{k,m}[j] \nonumber\\
&\quad
+ \sqrt{\frac{\INRP^2-\SNRP}{\INRP}} \bx_{k+1,h}[j] \nonumber\\
&\quad+ \sqrt{\frac{\SNRP}{\INRP}}\bx_{k,l}[j] + \sqrt{\frac{\SNRP-\INRP}{\INRP}}\bx_{k+1,m}[j] \nonumber\\
&\quad+ \bx_{k+1,l}[j]+\bz_k[k],
\end{align}
it decodes $\bx_{k,h}[j]$, $\bx_{k,m}[j]$, and $\bx_{k+1,h}$ sequentially; that is in each step it treats everything else as noise, decodes the codeword, and removes the corresponding part from the received signal. This can be done as long as
\begin{align}
R_1&\leq \cgf{\SNRP+\INRP+1}{\frac{\SNRP^2}{\INRP^2}+\INRP+1}\\
R_2&\leq \cgf{\frac{\SNRP^2}{\INRP^2}+\INRP+1}{\frac{\SNRP}{\INRP}+\INRP+1}\\
R_1&\leq \cgf{\frac{\SNRP}{\INRP}+\INRP+1}{2\frac{\SNRP}{\INRP}+1},
\end{align}
which are all satisfied with the rates in \eqref{eq:rates-w}. The remaining part of the signal would be
\begin{align}
\tilde{\by}_k[j]&=  \sqrt{\frac{\SNRP}{\INRP}}\bx_{k,l}[j] +\sqrt{\frac{\SNRP-\INRP}{\INRP}}\bx_{k+1,m}[j] \nonumber\\
&\quad+ \bx_{k+1lm}[j]+\bz_k[k],
\end{align}
which will be stored for further processing and for decoding $\bx_{k,l}[j]$ later. In the next block, once $\bx_{k,m}[j+1]$ is decoded, the decoder again recalls $\tilde{\by}_k[j]$ and subtracts from it the part corresponding to $\bx_{k,m}[j+1]=\bx_{k+1,m}[j]$. Therefore, it obtains
\begin{align}
&\tilde{\by}_k[j]-\sqrt{\frac{\SNRP-\INRP}{\INRP}}\bx_{k+1,m}[j] \nonumber\\
&\quad= \sqrt{\frac{\SNRP}{\INRP}}\bx_{k,l}[j]+ \bx_{k+1,l}[j]+\bz_k[k],
\end{align}
from which $\bx_{k,l}[j]$ can be decoded as long as
\begin{align}
R_3 \leq \cgf{\frac{\SNRP}{\INRP}+2}{2},
\end{align}
which clearly holds with $R_3$ in \eqref{eq:rates-w}.
Therefore, in each block, the receiver can decode the low power codeword of the previous block after removing the moderate power interfering signal. However, this process does not have to continue forever, since in the $K$-th block, the moderate power interfering signal at receiver $k$ would be
\[
\bx_{k+1,m}[K]=\bx_{k+2,m}[K-1]=\cdots=\bx_{k+K,m}[1]=\bx_{k,m}[1],
\]
which was already decoded in the first block.

In summary, the total rate can be achieved per user per block would be
\begin{align}
&R_{\mathrm{sym}}\nonumber\\
&=\frac{K R_1 + R_2 + K R_3}{K}\\
&= \cg{\INRP} +\frac{1}{2}\log\left(\frac{\frac{\SNR}{\INR} +2}{2\frac{\SNR}{\INR}+1}\right)   \nonumber\\&\qquad+\frac{1}{K} \left[\frac{1}{2}\log(1+\SNR^2)-\frac{1}{2}\log(1+\INR^3)\right]\\
&\geq \cg{\INRP} +\frac{1}{2}\log\frac{1}{2}\nonumber\\&\qquad+\frac{1}{K} \left[\frac{1}{2}\log(1+\SNR)^2-\frac{1}{2}-\frac{1}{2}\log(1+\INR)^3\right] \\
&= E+ \frac{(2D-3E)}{K} -\frac{1}{2} -\frac{1}{2K}.
\end{align}

\subsubsection{Upper Bound}
In this regime, we use the same upper bound as in the case of very weak interference regime.
In particular, from (\ref{Twousages}), we have
\begin{align}
\mathcal{C}_{\mathrm{sym},\mathrm{G}}^{\mathrm{FB}}(K)&\leq (B-E)+ \frac{(A+C+E-2B)}{K}\\
&\leq 1+ E + \frac{(A+C+E-4E)}{K}\\
&\leq 1+ E + \frac{(1+D+1/2+D+E-4E)}{K}\label{Weakmiddle}\\
&= E+\frac{(2D-3E)}{K} + 1+ \frac{3}{2K}.
\end{align}
Here, in (\ref{Weakmiddle}), we have used the fact that for the weak interference regime, we have $2E\leq B\leq 1+2E$, $A\leq 1+D$, and $C\leq 1/2 + D$.
Therefore, we have
\begin{align}
&\Bigg[E+ \frac{(2D-3E)}{K}\Bigg] -\frac{1}{2} -\frac{1}{2K}\nonumber\\
&\hspace{1cm}\leq    \mathcal{C}_{\mathrm{sym},\mathrm{G}}^{\mathrm{FB}}(K)\nonumber\\
&\hspace{1cm}\leq \Bigg[E+\frac{(2D-3E)}{K}\Bigg] + 1+ \frac{3}{2K},
\end{align}
which implies that the gap is bounded as follows:
\begin{align}
\Delta&\leq \frac{3}{2}+ \frac{2}{K}
\end{align}
which is at most $3$ bits per user-pair and we have
\begin{align}
\DOF^{\mathrm{FB}}(\alpha,K)&= \alpha+ \frac{(2-3\alpha)}{K}, \hspace{3pt} 1/2\leq \alpha\leq 2/3.
\end{align}

\subsection{Moderate Interference: $2/3\leq \alpha\leq 1$}
\subsubsection{Coding Scheme}
In this regime we use the private and common message for the encoding scheme. Assume each transmitter has two messages, namely the high power (common) message $M_k^{(h)}$, and the low power (private) message $M_k^{(l)}$. The following transmission scheme aims to convey the common message $M_{k}^{(h)}$ to both receivers $k$ and $k+1$. However, the private message $M_{k}^{(l)}$ can be decoded only by the respective receiver.

We assume that the high power messages of all users have the same rate. Similarly, the rate of the low power messages for all users are the same; that is
\begin{align}
R_1&= \frac{\log |\cM_k^{(h)}|}{n}, \qquad k=1,\dots,K,\nonumber\\
R_2&= \frac{\log |\cM_k^{(l)}|}{n}, \qquad k=1,\dots,K.
\end{align}

The encoder first maps its messages to Gaussian codewords with unit average power, $\bx_{k,h}$ and $\bx_{k,l}$, and sends
\[
\bx_k= \sqrt{\frac{\INRP-1}{\INRP}} \bx_{k,h}  + \sqrt{\frac{1}{\INRP}} \bx_{k,l}.
\]
The receiver node $k$, upon receiving $\by_k$, with
\begin{align}
&\by_k[j] \nonumber\\
&=  \sqrt{\SNRP} \bx_k[j]+ \sqrt{\INRP} \bx_{k+1}[j] +\bz_k[j]\\
&=   \sqrt{ \frac{\SNRP}{\INRP} (\INRP-1)} \bx_{k,h} [j] +  \sqrt{ \INRP-1} \bx_{k+1,h}[j]\nonumber\\
&\qquad+\sqrt{ \frac{\SNRP}{\INRP} } \bx_{k,l} [j]
		+  \bx_{k+1,l}[j] + \bz_k[j],
\end{align}
first jointly decodes the codewords $\bx_{k,h}^{(h)}$ and $\bx_{k+1,h}^{(h)}$ treating all the rest as noise. Here we deal with a multiple access channel, whose achievable rate is characterized by
\begin{align}
R_1 &\leq \cgf{\SNRP+2}{\frac{\SNRP}{\INRP} +2},\\
R_1 &\leq \cgf{\INRP+ \frac{\SNRP}{\INRP} +1}{\frac{\SNRP}{\INRP} +2},\\
2R_1 &\leq \cgf{\SNRP+\INRP+ 1}{\frac{\SNRP}{\INRP} +2}.
\end{align}
In particular, it is easy to show that\footnote{We assume $\INR\geq 1$. If this condition does not hold, then a bounded gap result can be shown in a similar fashion to Appendix~\ref{app: excluded-range}.}
\begin{align}
R_1=\frac{1}{4}\log^+\left(\frac{\SNRP+\INRP+ 1}{\frac{\SNRP}{\INRP} +2}\right)
\end{align}
satisfies the above constraints. After decoding the high power codewords, and removing them from the received signal, receiver $k$ decodes its own low power message by treating the other private codeword as noise. This private message can be reliably decoded provided that
\begin{align}
R_2 &\leq \cgf{\frac{\SNRP}{\INRP} +2}{2},
\end{align}
which yields an achievable total rate of
\begin{align}
R_{\mathrm{sym}}&=R_1+R_2\\
&=\frac{1}{4}\log\left(\SNRP+\INRP+ 1\right) + \frac{1}{4}\log\left(2+\frac{\SNR}{\INR}\right)-\frac{1}{2}\\
&\geq \frac{1}{2}\log(1+\SNR)-\frac{1}{4}\log(1+\INR) -\frac{1}{2}\\
&= D-\frac{E}{2}-\frac{1}{2}.
\end{align}
\subsubsection{Upper Bound}
In this regime, we will develop a different upper bound that is analogous
to the type-I upper bound obtained for the linear deterministic channel model.
We have the following bound on the feedback sum capacity:
\begin{align}
\mathcal{C}_{\mathrm{sum},\mathrm{G}}^{\mathrm{FB}}(K)&\leq \frac{K}{2}(A+C-E).\label{GaussianTypeIBound}
\end{align}
The proof of (\ref{GaussianTypeIBound}) is given in the appendix.
Hence, (\ref{GaussianTypeIBound}) implies that the symmetric feedback capacity is upper bounded as
\begin{align}
\mathcal{C}_{\mathrm{sym},\mathrm{G}}^{\mathrm{FB}}(K)&\leq \frac{A+C-E}{2}\\
&= \frac{A+C}{2}- \frac{E}{2}.
\end{align}
Note that in this regime, we have
\begin{align}
2A&= \log(1+\SNR+\INR+2\sqrt{\SNR\cdot\INR})\\
&\leq \log(1+4\SNR)\\
&\leq \log(4)+\log(1+\SNR)\\
&= 2+ 2D,
\end{align}
and similarly,
\begin{align}
2C&= \log(1+\SNR+\INR)\\
&\leq \log(1+2\SNR)\\
&\leq \log(2)+\log(1+\SNR)\\
&= 1+ 2D,
\end{align}
which implies that
\begin{align}
\frac{A+C}{2}&\leq D+\frac{3}{4}.
\end{align}
Hence, we have
\begin{align}
\Bigg[D-\frac{E}{2}\Bigg]-\frac{1}{2}\leq \mathcal{C}_{\mathrm{sym},\mathrm{G}}^{\mathrm{FB}}(K)\leq \Bigg[D-\frac{E}{2}\Bigg]+\frac{3}{4},
\end{align}
so that the gap between the upper and lower bounds is at most $5/4$ bits and we have
\begin{align}
\DOF^{\mathrm{FB}}(\alpha,K)&= 1-\frac{\alpha}{2}, \quad 2/3\leq \alpha\leq 1.
\end{align}

\subsection{Strong Interference: $1\leq \alpha\leq 2$}
\subsubsection{Coding Scheme}
In this regime, we have $\SNRP \leq \INRP \leq \SNRP^2$.
The encoding scheme for this interference regime is simple, and the desired degrees of freedom can be achieved in one block. Denote the message of user $k$ by $M_k\in \cM_k$, where all the message sets have the same size which results in a symmetric rate  of $R$.  Each user takes a random Gaussian codebook with rate $R$ and unit average power. Then it randomly maps its message to $\bx_k$ and sends it over the channel. The $k$-th receiver observes $\by_k$ through a multiple access channel from the $k$-th and $(k+1)$-th  transmitters, in which it has to decode both messages. The achievable rate of the MAC is well-known as \cite{Cover:book}
\begin{align}
R &\leq \cg{\SNRP},\\
R &\leq \cg{\INRP},\\
2R &\leq \cg{\INRP+\SNRP}.
\end{align}
Hence, it is clear that by choosing
\begin{align}
R_{\mathrm{sym}}&=\frac{1}{4} \log\left(1+\INRP+\SNRP\right)\\
&= \frac{C}{2},
\end{align}
all constraints are satisfied and a symmetric rate of $R_{\mathrm{sym}}$ is therefore achievable.

\subsubsection{Upper bound}
For this regime, we use the same upper bound developed in the previous section:
\begin{align}
\mathcal{C}_{\mathrm{sym},\mathrm{G}}^{\mathrm{FB}}(K)&\leq \frac{A+C-E}{2}.
\end{align}
Therefore, the symmetric feedback capacity satisfies
\begin{align}
\frac{C}{2}&\leq \mathcal{C}_{\mathrm{sym},\mathrm{G}}^{\mathrm{FB}}(K)\leq \frac{C}{2}+\frac{A-E}{2}\label{BoundS},
\end{align}
and the gap between the bounds is
\begin{align}
\Delta&= \frac{A-E}{2}\\
&= \frac{\log(1+\SNR+\INR+2\sqrt{\SNR\cdot\INR})-\log(1+\INR)}{4}\\
&\leq \frac{\log(1+4\INR)-\log(1+\INR)}{4}\\
&\leq \frac{\log(4)+\log(1+\INR)-\log(1+\INR)}{4}\\
&=\frac{1}{2}.
\end{align}
Moreover, from (\ref{BoundS}), it is straightforward to show that
\begin{align}
\DOF^{\mathrm{FB}}(\alpha,K)&= \frac{\alpha}{2}, \quad 1\leq \alpha\leq 2.
\end{align}

\subsection{Very-Strong Interference: $\alpha\geq 2$}
\subsubsection{Coding Scheme}
The encoding scheme we propose here takes $K$ blocks, each of length $T$. We assume each user $k$ has a total of $K+1$ messages, namely $\left( M_k^{(l)}, M_k^{(1,h)}, \dots, M_k^{(K,h)} \right)$, that it wishes to send to its respective receiver over $K$ transmission blocks. We assume that  $M_k^{(l)}\in \cM_k^{(l)}$ and $M_k^{(j,h)} \in \cM_k^{(j,h)}$, where $\cM$'s are the  message sets. Moreover, we set the size of these message sets so that\footnote{Here we assume $\SNR\geq 2$ to make sure that $R_2$ is positive. Note that $R_1$ is positive since $\INR\geq \SNR^2$. A similar analysis as in Appendix~\ref{app: excluded-range} can be done if $\SNR<2$.}
\begin{align}
\begin{split}
\frac{\log |\cM_k^{(l)}| }{n} &= R_1 = \cgfp{\frac{\INRP}{\SNRP^2}+1}{2},
\\&\qquad k=1,\dots,K \\
\frac{\log |\cM_k^{(j,h)}| }{n} &= R_2=\cgfp{\SNRP+1}{3},
\\&\qquad k=1,\dots,K, \ j=1,\dots,K.
\end{split}
\label{eq:rates-VS}
\end{align}

That is, all the first messages of all the users have the same rate. Furthermore, the rates of all the remaining messages are also identical.
Each message is encoded by a capacity achieving Gaussian codebook with unit variance. Hence user $k$ has $K+1$ Gaussian codewords,
$\bs_k^{(l)}, \bs_k^{(1,h)}, \dots, \bs_k^{(K,h)}$, each of length $n$.

The signal sent by  transmitter $k$ in block $j$ is composed of two parts, the high power part $\bx_{k,h}[j]$ and low power $\bx_{k,l}[j]$:
\begin{align}
\bx_k[j]=  \sqrt{ \frac{\SNRP -1}{\SNRP} } \bx_{k,h}[j] + \sqrt{ \frac{1}{\SNRP} }  \bx_{k,l}[j].
\end{align}
 In all blocks, the high power part is the codeword corresponding to a fresh message. In the first block, since the nodes have not yet received any feedback, their low power codewords also describe fresh messages. However, for all blocks $j\geq 2$, the low level codeword used to form the transmitting signal is the low level codeword of their neighbor sent on the previous block. We will show that this message can be decoded from the signal received over the feedback link at the end of the last block. More precisely,
 \begin{align}
 \bx_{k,h}[j]&= \bs_k^{(j,h)}, \quad k=1,2,\dots,K,\ j=1,2,\dots,K,\\
 \bx_{k,l}[1]&= \bs_k^{(l)}, \quad k=1,2,\dots,K,\\
  \bx_{k,l}[j]&= \bx_{k+1,l}[j-1], \quad k=1,2,\dots,K,\ j=2,\dots,K.
 \end{align}
Therefore we have the following recursive relationship between the low power codewords:
\begin{align}
 \bx_{k,l}[j]&= \bx_{k+1,l}[j-1]=\bx_{k+2,l}[j-2] =\cdots = \bx_{k+j-1,l}[1]\nonumber\\
 &=\bs_{k+j-1}^{(l)}\label{relationVS},
\end{align}
where all the user and block indicators are modulo $K$, e.g., $\bx_{k+j-1}[1]=\bx_{(k+j-1) \mod K}[1]$.

\paragraph{Decoding the feedback signal at encoder}
As stated above, in order to form the transmitting signal in block $j$, transmitter $k$ uses the low power codeword sent by user $k+1$ in block $j-1$. We first show that this codeword can be decoded based on the signal it receives over the feedback link at the end of block $j-1$.

Once $\by_{k}[j-1]$ is received, transmitter $k$ first removes its own signal, $\bx_k[j-1]$, to obtain
\begin{align}
&\by_k[j-1]-\sqrt{\SNRP}\bx_k[j-1]\nonumber\\
&=\sqrt{\INRP} \bx_{k+1}[j-1]+  \bz_k[j-1]\\
&=\sqrt{ \frac{\INRP}{\SNRP} (\SNRP-1)} \bx_{k+1,h}[j-1] +  \sqrt{ \frac{\INRP}{\SNRP} } \bx_{k+1,l}[j-1]\nonumber\\&\qquad+\bz_k[j-1],
\end{align}
from which it has to decode both $\bx_{k+1,h}[j-1]$ and $\bx_{k+1,l}[j-1]$. It first decodes $\bx_{k+1,h}[j-1]$, treating everything else as noise. Then, it removes $\bx_{k+1,h}[j-1]$ from the signal and decodes $\bx_{k+1,l}[j-1]$ in a similar manner. This is possible as long as
\begin{align}
R_2 &\leq \cgf{\INRP+1}{\frac{\INRP}{\SNRP} +1}\\
R_1 &\leq \cg{\frac{\INRP}{\SNRP}},
\end{align}
which are clearly satisfied by the rates chosen in \eqref{eq:rates-VS}. Therefore, the transmitter $k$ has access to $\bx_{k+1,l}[j-1]$, which will be used as its low power codeword for block $j$.

\paragraph{Decoding process at the receiver}

The signal sent by user $k$ over the $j$-th block is given by
\begin{align}
\bx_k[j]=  \sqrt{ \frac{\SNRP -1}{\SNRP} } \bx_{k,h}[j] + \sqrt{ \frac{1}{\SNRP} }  \bx_{k,l}[j],
\end{align}
which results in
\begin{align}
\by_k[j] &=  \sqrt{\SNRP} \bx_k[j]+ \sqrt{\INRP} \bx_{k+1}[j] +\bz_k[j]\\
&=   \sqrt{ \frac{\INRP}{\SNRP} (\SNRP-1)} \bx_{k+1,h}[j] +  \sqrt{ \frac{\INRP}{\SNRP} } \bx_{k+1,l}[j]\nonumber\\
&	\qquad	+  \sqrt{\SNRP-1}\bx_{k,h}[j] +\bx_{k,l}[j]+ \bz_k[j].
\end{align}
At the end of the $j$-th block, user $k$  sequentially decodes the codewords $\bx_{k+1,h}[j]$, $\bx_{k+1,l}[j]$, and $\bx_{k,h}[j]$. At each step, it decodes the corresponding codewords, treating all the remaining parts as noise. Once one codeword is decoded, it removes it from its received signal, and proceeds with the next codeword. This can be done provided that
\begin{align}
R_2 &\leq \cgf{\INRP+\SNRP+1}{\frac{\INRP}{\SNRP} + \SNRP +1},\\
R_1 &\leq \cgf{\frac{\INRP}{\SNRP} + \SNRP +1}{\SNRP +1},\\
R_2 &\leq \cgf{\SNRP +1}{2}.
\end{align}
It is easy to check that all constraints are satisfied by the choice of $R_1$ and $R_2$ in \eqref{eq:rates-VS}.

At the end of each block, each receiver can decode its respective high power codeword, as well as some high power and low power codewords from other users which  it is not intended to decode. However, from (\ref{relationVS}), the low power codeword decoded at receiver $k$ at the very last block would be
\[
\bx_{k+1,l}[K]=\bx_{k+2,l}[K-1]=\cdots=\bx_{k+K,l}[1] \stackrel{(*)}{=} \bx_{k,l}[1]=\bs_k^{(l)},
\]
where $(*)$ holds since $k+K=k\ \mod K$. Therefore all the intended messages  for receiver $k$, can be decoded using this scheme in $K$ blocks.
The total rate of communication would be
\begin{align}
R_{\mathrm{sym}} &=\frac{R_1+ K R_2}{K} = \cg{\SNRP} \nonumber\\&\qquad+ \frac{1}{2K} \log \left(\frac{\INRP}{\SNRP^2}+1\right) -\frac{K\log 3+1}{2K}\\
&= D+  \frac{1}{2K} \log \left(\frac{\INRP}{\SNRP^2}+1\right) -\frac{K\log 3+1}{2K}\\
&\geq D+  \frac{1}{2K} \log \left(\frac{1+\INRP}{(1+\SNRP)^2}\right) -\frac{K\log 3+1}{2K}\\
&= D+  \frac{(E-2D)}{K}- \frac{(K\log 3+1)}{2K}.\label{lowerVSA}
\end{align}
\subsubsection{Upper Bound}
For this regime we use the following upper bound from Theorem \ref{theoremSUMUB}, similar to the LD case:
\begin{align}
&\mathcal{C}_{\mathrm{sum},\mathrm{G}}^{\mathrm{FB}}(K)\nonumber\\
&\leq \max_{p(x_{1},\ldots,x_{K})} \Big[h(Y_{K})+h(Y_{K-1}|X_{K},Y_{K})+\ldots\nonumber\\
&\hspace{2.2cm}+h(Y_{1}|X_{2},Y_{2},\ldots,X_{K},Y_{K})\nonumber\\
&\hspace{2.2cm}- h(Y_{1},\ldots,Y_{K}|X_{1},\ldots,X_{K})\Big]\label{VSGaussian}.
\end{align}
We upper bound the first term in (\ref{VSGaussian}) as
\begin{align}
h(Y_{K})&\leq \frac{1}{2}\log\left(1+\SNR+\INR+2\sqrt{\SNR\cdot\INR}\right) + c\\
&= A+c,
\end{align}
where $c=(1/2)\log(2\pi\mbox{e})$.

For any $2\leq k \leq (K-1)$, we bound
\begin{align}
h(Y_{k}|X_{k+1},Y_{k+1},\ldots,X_{K},Y_{K})&\leq h(Y_{k}|X_{k+1})\\
&= h(\sqrt{\SNR}X_{k}+Z_{k}|X_{k+1})\\
&\leq h(\sqrt{\SNR}X_{k}+Z_{k})\\
&\leq \frac{1}{2}\log\left(1+\SNR\right)+c\\
&= D+c.
\end{align}
Finally, we bound the penultimate term in (\ref{VSGaussian}) as follows:
\begin{align}
&h(Y_{1}|X_{2},Y_{2},\ldots,X_{K},Y_{K})\nonumber\\
&\leq  h(Y_{1}|X_{2},Y_{2},X_{K},Y_{K})\\
&= h(\sqrt{\SNR}X_{1}+Z_{1}|X_{2},Y_{2},X_{K},Y_{K})\\
&\leq h(\sqrt{\SNR}X_{1}+Z_{1}|\sqrt{\INR}X_{1}+Z_{K})\label{VSmiddleB}\\
&\leq \frac{1}{2}\log\left(\frac{1+\SNR+\INR}{1+\INR}\right)+c\\
&= C-E+c,
\end{align}
where in (\ref{VSmiddleB}), we used the fact that $Y_{K}=\sqrt{\SNR}X_{K}+\sqrt{\INR}X_{1}+Z_{K}$ and the fact that conditioning reduces differential entropy.
Finally, we note that
\begin{align}
h(Y_{1},\ldots,Y_{K}|X_{1},\ldots,X_{K})&= h(Z_{1},\ldots,Z_{K}|X_{1},\ldots,X_{K})\\
&= h(Z_{1},\ldots,Z_{K})\\
&= \sum_{k=1}^{K}h(Z_{k})\\
&= Kc.
\end{align}
Hence, the feedback sum capacity is upper bounded as follows:
\begin{align}
\mathcal{C}_{\mathrm{sum},\mathrm{G}}^{\mathrm{FB}}(K)&\leq A+ (K-2)D + C-E,
\end{align}
which implies that the symmetric feedback capacity satisfies
\begin{align}
\mathcal{C}_{\mathrm{sym},\mathrm{G}}^{\mathrm{FB}}(K)&\leq D+ \frac{(A+C-2D-E)}{K}\label{upperVSA}.
\end{align}
We now simplify this upper bound to compare it with the lower bound obtained in (\ref{lowerVSA}).

We note that
\begin{align}
2A&=\log(1+\SNR+\INR+2\sqrt{\SNR\cdot\INR})\\
&\leq \log(1+4\INR)\\
&\leq \log(4)+\log(1+\INR)\\
&= 2 + 2E,
\end{align}
and
\begin{align}
2C&=\log(1+\SNR+\INR)\\
&\leq \log(1+2\INR)\\
&\leq \log(2)+\log(1+\INR)\\
&= 1 + 2E,
\end{align}
which together imply that
\begin{align}
A+C-2D-E&\leq (3/2)+2E-2D-E\\
&= (E-2D)+ 3/2.
\end{align}
Hence, from (\ref{lowerVSA}) and (\ref{upperVSA}), the symmetric feedback capacity satisfies
\begin{align}
&\Bigg[D+  \frac{(E-2D)}{K}\Bigg]- \frac{(K\log(3)+1)}{2K}\nonumber\\
&\hspace{1.5cm}\leq \mathcal{C}_{\mathrm{sym},\mathrm{G}}^{\mathrm{FB}}(K)\nonumber\\
&\hspace{1.5cm}\leq \Bigg[D+  \frac{(E-2D)}{K}\Bigg] + \frac{3}{2K}\label{BoundVS}
\end{align}
which implies that the gap is given as
\begin{align}
\Delta&= \frac{4+K\log(3)}{2K}\\
&\leq \frac{2+K}{K}
\end{align}
which is at most $2$-bits. We note here that the gap of $2$-bits can be reduced further to $1$-bit by modifying the power allocation
in our coding scheme. The resulting gap analysis is however complicated and is not pursued here.

Moreover, from (\ref{BoundVS}), we have
\begin{align}
\DOF^{\mathrm{FB}}(\alpha,K)&= 1+\frac{(\alpha-2)}{K}, \quad \alpha\geq 2.
\end{align}
\section{Conclusions}\label{CONCLUSION}
In this paper, we have considered the $K$-user cyclic Z-interference channel with noiseless feedback.
The symmetric feedback capacity of the linear deterministic CZIC has been completely characterized for all
interference regimes. Using insights from the linear model, the symmetric feedback capacity for the
Gaussian CZIC has been characterized within a constant number of bits for all interference regimes.
As a consequence of the constant bit gap result, the symmetric feedback degrees of freedom for the Gaussian CZIC
have also been characterized.
It has been shown that the capacity gain obtained via feedback decreases as the number of users increases.
The resulting $\DOF^{\mathrm{FB}}(\alpha,K)$ for $K>2$ users is a skewed $V$-curve, as a function of the interference parameter $\alpha$.
Moreover as $K\rightarrow \infty$, the resulting skewed $V$-curve converges to the well known $W$-curve
corresponding to the no-feedback $\DOF$. As part of future work, we believe that  the characterization of the approximate feedback capacity region of the $K$-user Gaussian CZIC is an interesting problem. Moreover, we believe that new coding schemes and novel upper bounds would be required to achieve this goal.

\section{Appendix}
\subsection{Proof of Theorem \ref{theoremK2}}
We show that the normalized symmetric feedback capacity is upper bounded as follows:
\begin{align}
\mathcal{C}^{\mathrm{FB}}_{\mathrm{sym},\mathrm{LD}}(\alpha,K)&\leq \max\left(1-\frac{\alpha}{2},\frac{\alpha}{2}\right)\label{UBK2}.
\end{align}
To prove (\ref{UBK2}), we first prove the following upper bound on the sum of the rates of users $1$ and $2$:
\begin{align}
&T(R_{1}+R_{2})\nonumber\\
&= H(W_{1})+H(W_{2})\\
&= H(W_{1}|W_{3},\ldots,W_{K}) + H(W_{2}|W_{1},W_{3},\ldots,W_{K})\label{T2a}\\
&\leq I(W_{1};Y_{1}^{T}|W_{3},\ldots,W_{K}) \nonumber\\
&\qquad+ I(W_{2};Y_{2}^{T},Y_{1}^{T}|W_{1},W_{3},\ldots,W_{K})+\epsilon_{T}\label{T2b}\\
&= I(W_{1};Y_{1}^{T}|W_{3},\ldots,W_{K}) \nonumber\\
&\qquad+ H(Y_{2}^{T},Y_{1}^{T}|W_{1},W_{3},\ldots,W_{K})+\epsilon_{T}\label{T2c}
\end{align}
\begin{align}
&= H(Y_{1}^{T}|W_{3},\ldots,W_{K}) \nonumber\\
&\qquad+ H(Y_{2}^{T}|Y_{1}^{T},W_{1},W_{3},\ldots,W_{K})+\epsilon_{T}\\
&\leq H(Y_{1}^{T}) + H(Y_{2}^{T}|Y_{1}^{T},W_{1},W_{3},\ldots,W_{K})+\epsilon_{T}\label{T2d}\\
&\leq T\max(m,n) + H(Y_{2}^{T}|Y_{1}^{T},W_{1},W_{3},\ldots,W_{K})+\epsilon_{T}\label{T2e}\\
&\leq T\max(m,n) + T(n-m)^{+}+\epsilon_{T}\label{T2f},
\end{align}
where (\ref{T2a}) follows from the fact that the messages $(W_{1},\ldots,W_{K})$ are all mutually independent, (\ref{T2b}) follows from Fano's inequality \cite{Cover:book}, (\ref{T2c}) follows from the deterministic nature of the channel model and (\ref{T2d}) follows from the fact that conditioning reduces entropy.

Before proving (\ref{T2f}) we first prove the following claim:
\begin{claim}\label{claim1}
$(X_{1t},X_{3t},\ldots,X_{Kt})$ is a deterministic function of
$(Y_{1}^{t-1},W_{1},W_{3},\ldots,W_{K})$.
\end{claim}
\begin{proof}
First note that from (\ref{encodingfunction}), we have
\begin{align}
X_{1t}=f_{1t}\left(W_{1},Y_{1}^{t-1}\right),
\end{align}
and
\begin{align}
X_{Kt}&=f_{Kt}\left(W_{K},Y_{K}^{t-1}\right)\\
&= f_{Kt}\left(W_{K},X_{1}^{t-1},X_{K}^{t-1}\right),
\end{align}
which together imply that
\begin{align}
(X_{1t},X_{Kt},X_{1}^{t-1},X_{K}^{t-1})=f\left(W_{1},W_{K},Y_{1}^{t-1}\right).
\end{align}
Repeating this argument for $k=K-1,\ldots,3$, the proof of the claim is straightforward.
\end{proof}
We now bound the second term in (\ref{T2e}) as follows:
\begin{align}
&H(Y_{2}^{T}|Y_{1}^{T},W_{1},W_{3},\ldots,W_{K})\nonumber\\
&\leq \sum_{t=1}^{T}H(Y_{2t}|Y_{1t},W_{1},W_{3},\ldots,W_{K},Y_{1}^{t-1})\\
&= \sum_{t=1}^{T}H(Y_{2t}|Y_{1t},X_{1t},W_{1},X_{3t},W_{3},\ldots,X_{Kt},W_{K},Y_{1}^{t-1})\label{T2main}\\
&\leq \sum_{t=1}^{T}H(Y_{2t}|Y_{1t},X_{1t},X_{3t})\\
&\leq \sum_{t=1}^{T}H(X_{2t}|Y_{1t},X_{1t},X_{3t})\\
&\leq T(n-m)^{+}\label{T2main2},
\end{align}
where (\ref{T2main}) follows from Claim \ref{claim1} and (\ref{T2main2}) follows from the fact that $(X_{1t},Y_{1t})$ completely determine
at least $m$ levels of $X_{2t}$. This completes the proof of (\ref{T2f}).
Dividing (\ref{T2f}) by $nT$ and taking the limit $T\rightarrow \infty$, we have $\epsilon_{T}\rightarrow 0$, which yields
\begin{align}
\frac{R_{1}+R_{2}}{n}&\leq \max\left(\frac{m}{n},1\right)+ \bigg(1-\frac{m}{n}\bigg)^{+}\\
&= \max(\alpha,1)+ (1-\alpha)^{+}\\
&= \max\left(2-\alpha,\alpha\right).
\end{align}
In a similar manner it can be shown that for any $1\leq j\leq K$,
\begin{align}
\frac{R_{j}+R_{(j+1)\mbox{mod}(K)}}{n}&\leq  \max\left(2-\alpha,\alpha\right).
\end{align}
Adding all such $K$ upper bounds, we obtain
\begin{align}
\frac{2(R_{1}+\ldots+R_{K})}{n}&\leq K\max\left(2-\alpha,\alpha\right),
\end{align}
and hence,
\begin{align}
\mathcal{C}^{\mathrm{FB}}_{\mathrm{sym},\mathrm{LD}}(\alpha,K)&\leq \max\left(1-\frac{\alpha}{2},\frac{\alpha}{2}\right).
\end{align}
This upper bound on the normalized symmetric feedback capacity is independent of $K$ and is
the same as the normalized symmetric capacity \emph{without feedback} when $\alpha\in [2/3,2]$.
Hence, for this interference regime, feedback does not increase the symmetric capacity.
Also note that the range of $\alpha$ in deriving these bounds is immaterial and hence from a
symmetric feedback capacity point of view, the feedback capacity for $K=2$ users always serves as an upper bound for any $K>2$.

\subsection{Proof of Theorem \ref{theoremSUMUB}}
In this section we provide the proof for Theorem \ref{theoremSUMUB}
for the special case in which $K=3$ and $\pi=(1,2,3)$ (identity
permutation). The generalization to arbitrary $(K, \pi)$ is
straightforward.

For the $3$-user interference channel with local feedback, we have the
following upper bound on the sum-rate:
\begin{align}
&T\left(R_{1}+R_{2}+R_{3}\right)\nonumber\\
&= H(W_{1})+H(W_{2})+H(W_{3})\\
&= H(W_{1})+ H(W_{2}|W_{1})+ H(W_{3}|W_{1},W_{2})\label{T3aK3}\\
&\leq I(W_{1};Y_{1}^{T})+ I(W_{2};Y_{2}^{T},Y_{1}^{T}|W_{1})\nonumber\\&\quad+
I(W_{3};Y_{3}^{T},Y_{2}^{T},Y_{1}^{T}|W_{1},W_{2}) + \epsilon_{T}\label{T3bK3}\\
&= H(Y_{1}^{T}) + H(Y_{2}^{T},Y_{1}^{T}|W_{1})+ H(Y_{3}^{T},Y_{2}^{T},Y_{1}^{T}|W_{1},W_{2})\nonumber\\
&\quad -H(Y_{1}^{T}|W_{1}) - H(Y_{2}^{T},Y_{1}^{T}|W_{1},W_{2})\nonumber\\&\quad-
H(Y_{3}^{T},Y_{2}^{T},Y_{1}^{T}|W_{1},W_{2},W_{3})+ \epsilon_{T}\\
&= H(Y_{1}^{T}) + H(Y_{2}^{T}|Y_{1}^{T},W_{1})+ H(Y_{3}^{T}|Y_{2}^{T},Y_{1}^{T},W_{1},W_{2})\nonumber\\
&\quad - H(Y_{3}^{T},Y_{2}^{T},Y_{1}^{T}|W_{1},W_{2},W_{3})+ \epsilon_{T}\label{T3cK3}\\
&\leq \sum_{t=1}^{T}\Big[H(Y_{1t}|Y_{1}^{t-1}) + H(Y_{2t}|Y_{1t},Y_{2}^{t-1},Y_{1}^{t-1},W_{1})\nonumber\\&\hspace{0.5cm}+ H(Y_{3t}|Y_{2t},Y_{1t},Y_{3}^{t-1},Y_{2}^{t-1},Y_{1}^{t-1},W_{1},W_{2})\nonumber\\
& \hspace{0.5cm}-
H(Y_{3t},Y_{2t},Y_{1t}|W_{1},W_{2},W_{3},Y_{1}^{t-1},Y_{2}^{t-1},Y_{3}^{t-1})\Big] + \epsilon_{T}\\
&\leq \sum_{t=1}^{T}\Big[H(Y_{1t}) + H(Y_{2t}|Y_{1t},X_{1t})\nonumber\\&\hspace{1.3cm}+ H(Y_{3t}|Y_{2t},Y_{1t},X_{2t},X_{1t})\nonumber\\&\hspace{1.3cm}-
H(Y_{3t},Y_{2t},Y_{1t}|X_{1t},X_{2t},X_{3t})\Big]+ \epsilon_{T}\\
&\leq T\max_{p(x_{1},x_{2},x_{3})}\Big[H(Y_{1}) + H(Y_{2}|Y_{1},X_{1})\nonumber\\&\hspace{2.5cm}+ H(Y_{3}|Y_{2},Y_{1},X_{2},X_{1})
\nonumber\\&\hspace{2.5cm}-H(Y_{3},Y_{2},Y_{1}|X_{1},X_{2},X_{3})\Big]+ \epsilon_{T},\label{T3dK3}
\end{align}
where (\ref{T3aK3}) follows from the independence of the messages, (\ref{T3bK3}) follows from Fano's inequality \cite{Cover:book}, and
(\ref{T3cK3}) follows from the fact that the negative term
corresponding to the $k$th mutual information is canceled by a part of the positive term
in the $(k+1)$th mutual information, for $k=1,\ldots,(K-1)$. Finally, dividing (\ref{T3dK3}) by $T$ and letting $T\rightarrow \infty$, we have the proof of Theorem \ref{theoremSUMUB}.

\subsection{Proof of (\ref{GaussianTypeIBound})}
We first obtain a bound on the sum of the rates of users $1$ and $2$:
\begin{align}
&T(R_{1}+R_{2})\nonumber\\
&= H(W_{1})+H(W_{2})\\
&= H(W_{1}|W_{3},\ldots,W_{K}) + H(W_{2}|W_{1},W_{3},\ldots,W_{K})\label{T2Ga}\\
&\leq I(W_{1};Y_{1}^{T},Z_{3}^{T},\ldots,Z_{K}^{T}|W_{3},\ldots,W_{K}) \nonumber\\&\quad+ I(W_{2};Y_{2}^{T},Y_{1}^{T},Z_{3}^{T},\ldots,Z_{K}^{T}|W_{1},W_{3},\ldots,W_{K})+\epsilon_{T}\label{T2Gb}\\
&= h(Y_{1}^{T},Z_{3}^{T},\ldots,Z_{K}^{T}|W_{3},\ldots,W_{K}) \nonumber\\&\quad+ h(Y_{2}^{T}|Y_{1}^{T},Z_{3}^{T},\ldots,Z_{K}^{T},W_{1},W_{3},\ldots,W_{K})\nonumber\\
&\quad -h(Y_{2}^{T},Y_{1}^{T},Z_{3}^{T},\ldots,Z_{K}^{T}|W_{1},W_{2},W_{3}\ldots,W_{K})+\epsilon_{T}\label{T2Gc}\\
&= h(Y_{1}^{T},Z_{3}^{T},\ldots,Z_{K}^{T}|W_{3},\ldots,W_{K}) \nonumber\\&\quad+ h(Y_{2}^{T}|Y_{1}^{T},Z_{3}^{T},\ldots,Z_{K}^{T},W_{1},W_{3},\ldots,W_{K})\nonumber\\
&\quad -\sum_{t=1}^{T}h(Y_{2t},Y_{1t},Z_{3t},\ldots,Z_{Kt}|W_{1},\ldots,W_{K},Y_{2}^{t-1},Y_{1}^{t-1},\nonumber\\&\hspace{5cm}Z_{3}^{t-1},\ldots,Z_{K}^{t-1})+\epsilon_{T}\label{T2Gcc}\\
&= h(Y_{1}^{T},Z_{3}^{T},\ldots,Z_{K}^{T}|W_{3},\ldots,W_{K}) +\epsilon_{T}\nonumber\\&\quad+ h(Y_{2}^{T}|Y_{1}^{T},Z_{3}^{T},\ldots,Z_{K}^{T},W_{1},W_{3},\ldots,W_{K})\nonumber\\
&\quad -\sum_{t=1}^{T}h(Y_{2t},Y_{1t},Z_{3t},\ldots,Z_{Kt}|X_{1t},X_{2t},X_{3t},\ldots,X_{Kt})\label{T2Gd}\\
&= h(Y_{1}^{T},Z_{3}^{T},\ldots,Z_{K}^{T}|W_{3},\ldots,W_{K})\nonumber\\&\quad + h(Y_{2}^{T}|Y_{1}^{T},Z_{3}^{T},\ldots,Z_{K}^{T},W_{1},W_{3},\ldots,W_{K})\nonumber\\
&\quad -\sum_{t=1}^{T}h(Z_{1t},Z_{2t},Z_{3t},\ldots,Z_{Kt})+\epsilon_{T}\label{T2Ge}\\
&\leq h(Y_{1}^{T},Z_{3}^{T},\ldots,Z_{K}^{T}) \nonumber\\&\quad+ h(Y_{2}^{T}|Y_{1}^{T},Z_{3}^{T},\ldots,Z_{K}^{T},W_{1},W_{3},\ldots,W_{K})\nonumber\\
&\quad -\sum_{t=1}^{T}h(Z_{1t},Z_{2t},Z_{3t},\ldots,Z_{Kt})+\epsilon_{T}\label{T2Gf}\\
&\leq h(Y_{1}^{T})+h(Z_{3}^{T},\ldots,Z_{K}^{T}) \nonumber\\&\quad+ h(Y_{2}^{T}|Y_{1}^{T},Z_{3}^{T},\ldots,Z_{K}^{T},W_{1},W_{3},\ldots,W_{K})\nonumber\\
&\quad -\sum_{t=1}^{T}h(Z_{1t},Z_{2t},Z_{3t},\ldots,Z_{Kt})+\epsilon_{T}\label{T2Gg}\\
&\leq TA+\sum_{t=1}^{T}h(Z_{3t},\ldots,Z_{Kt}) \nonumber\\&\quad+ h(Y_{2}^{T}|Y_{1}^{T},Z_{3}^{T},\ldots,Z_{K}^{T},W_{1},W_{3},\ldots,W_{K})\nonumber\\
&\quad -\sum_{t=1}^{T}h(Z_{1t},Z_{2t},Z_{3t},\ldots,Z_{Kt})+\epsilon_{T}\label{T2Gh}
\end{align}
\begin{align}
&\leq TA + \sum_{t=1}^{T}h(Y_{2t}|X_{1t},Y_{1t},X_{3t}) -\sum_{t=1}^{T}h(Z_{2t})+\epsilon_{T}\label{T2Gi}\\
&= TA\nonumber\\&\quad +\sum_{t=1}^{T}h(\sqrt{\SNR}X_{2t}+Z_{2t}|X_{1t},\sqrt{\INR}X_{2t}+Z_{1t},X_{3t})\nonumber\\&\quad-\sum_{t=1}^{T}h(Z_{2t})+\epsilon_{T}\label{T2Gi2}\\
&\leq TA+ \sum_{t=1}^{T}h(\sqrt{\SNR}X_{2t}+Z_{2t}|\sqrt{\INR}X_{2t}+Z_{1t})\nonumber\\&\quad-\sum_{t=1}^{T}h(Z_{2t})
+\epsilon_{T}\label{T2Gi3}\\
&\leq TA+ T(C-E)+\epsilon_{T}\label{T2Gfa}\\
&= T(A+ C-E)+\epsilon_{T}\label{T2Gfb},
\end{align}
where (\ref{T2Ga}) follows from the independence of the messages, (\ref{T2Gb}) follows from Fano's inequality, (\ref{T2Gcc}) follows from the chain rule,
and (\ref{T2Gd}) follows from the following argument:
\begin{align}
&X_{1t} \mbox{ is a function of } (W_{1},Y_{1}^{t-1})\nonumber\\
&X_{2t} \mbox{ is a function of } (W_{2},Y_{2}^{t-1})\nonumber,\\
&X_{Kt} \mbox{ is a function of } (W_{K},X_{1}^{t-1},Z_{K}^{t-1})\nonumber,\\
&X_{(K-1)t} \mbox{ is a function of } (W_{K},X_{K}^{t-1},Z_{K-1}^{t-1})\nonumber,\\
&\hspace{0.1cm}\vdots\\
&X_{4t} \mbox{ is a function of } (W_{4},X_{5}^{t-1},Z_{4}^{t-1})\nonumber,\\
&X_{3t} \mbox{ is a function of } (W_{3},X_{4}^{t-1},Z_{3}^{t-1})\nonumber.
\end{align}
This argument allows us to write $(X_{1t},X_{2t},\ldots,X_{Kt})$ in the conditioning in the last term in (\ref{T2Gd}) and then use the memoryless
property of the channel to arrive at (\ref{T2Ge}).

The same argument also allows us to write $(Y_{1t},X_{1t},X_{3t})$ in the conditioning of the third term in (\ref{T2Gi}) and subsequently
drop all the remaining random variables from the conditioning. We remark here that this argument is similar to Claim \ref{claim1} used
in the proof of Theorem \ref{theoremK2} for the linear deterministic model.

Finally, normalizing (\ref{T2Gfb}) by $T$ and taking the limit $T\rightarrow \infty$, so that $\epsilon_{T}\rightarrow 0$, we have
\begin{align}
R_{1}+R_{2}&\leq A+C-E.
\end{align}
In a similar manner, it can be shown that for any $1\leq j\leq K$, we have
\begin{align}
R_{j}+R_{(j+1)}&\leq A+C-E.
\end{align}
Adding all such $K$ bounds, we obtain
\begin{align}
2(R_{1}+\ldots+R_{K})&\leq K(A+C-E),
\end{align}
which yields
\begin{align}
\mathcal{C}_{\mathrm{sum},\mathrm{G}}^{\mathrm{FB}}(K)&\leq \frac{K}{2}(A+C-E).
\end{align}
Hence, we have proved the analog of the type-I upper bound for the $K$-user Gaussian CZIC.

\subsection{Marginal Range of Parameters Excluded in Section~\ref{GaussianSection}}
\label{app: excluded-range}

In the coding scheme and performance analysis presented for various regimes of parameter in Section~\ref{GaussianSection}, we have inherently always assumed that $\SNRP$ and $\INRP$ (and possibly their ratio) are greater than certain constants, so that the desired rates are non-negative. Although this is a valid assumption for the range of parameters of primary interest, we prove the bounded gap from capacity result for arbitrary parameters for completeness. In this section we focus on the range of parameters excluded from the discussions in Section~\ref{GaussianSection}, and show the bounded gap result. In sake of brevity, we present this analysis only for the very weak interference regime ($0\leq \alpha \leq 1/2$). The analysis for other ranges of $\alpha$ is very similar, and is omitted.

\paragraph{Very Weak Interference $0\leq \alpha\leq 1/2$}
In this regime we have $\INRP^2\leq \SNRP$. Recall the rate allocation presented in \eqref{eq:rates-vw},
\begin{align}
\begin{split}
 R'_1&=\cgfp{\INRP+1}{3}, \quad k=1,\dots,K,  \\
R'_2&=\cgfp{\frac{\SNRP}{\INRP}+1}{2\INRP+1}, \quad k=1,\dots,K,\ j=1,\dots,K,\\
R'_3&= \cgfp{\INRP+1}{2}, \quad k=1,\dots,K, \ j=1,\dots,K.
\end{split}
\label{eq:rates-vw-p}
\end{align}
Consider the following four cases:
\begin{center}
  \begin{tabular}{| l || c | c | }
    \hline
    I &  $\INRP\geq 2$ &  $\SNRP\geq 2\INRP^2$ \\ \hline
II &  $\INRP\geq 2$ &  $\INRP^2 \leq \SNRP< 2\INRP^2$ \\ \hline
    III &  $\INRP< 2$ &  $\SNRP\geq 2\INRP^2$ \\ \hline
IV &  $\INRP< 2$ &  $\INRP^2 \leq \SNRP< 2\INRP^2$ \\ \hline
  \end{tabular}
\end{center}
\noindent \underline{\textbf{Case I.}} The conditions in the first case guarantee that all the rates in \eqref{eq:rates-vw} are positive, and so the analysis in Section~\ref{seubsec:G-VW} is valid.

In the following we analyze the remaining three cases which were excluded in Section~\ref{seubsec:G-VW}.

\noindent \underline{\textbf{Case II.}} In case II, $R'_2=0$, and hence, the total achievable rate would be
\begin{align}
R_{\mathrm{sym}}&=\frac{ R_1 +  K R_3}{K}\nonumber\\
&=  \frac{1}{2}\log(1+\INRP) +\frac{1}{2K}\log(1+\INRP) -\frac{K+\log 3}{2K}.
\end{align}
However, note that under the conditions of case II, from \eqref{eq:UB:G-VW} we have
\begin{align}
&\mathcal{C}_{\mathrm{sym},\mathrm{G}}^{\mathrm{FB}}(K)\nonumber\\
&\leq  (B-E) + \frac{E}{K}\nonumber\\
&\leq \frac{1}{2}\log \left(1+3\INRP^2+2\INRP+2\sqrt{2\INRP^3}\right)\nonumber\\&\quad- \frac{1}{2}\log(1+\INRP)+\frac{1}{2K}\log(1+\INRP)\nonumber\\
&\leq \frac{1}{2}\log 4\left(1+\INRP\right)^2 - \frac{1}{2}\log(1+\INRP)+\frac{1}{2K}\log(1+\INRP)\nonumber\\
&\leq \frac{1}{2}\log(1+\INRP) +\frac{1}{2K}\log(1+\INRP)+1.
\end{align}
Therefore, the gap between the upper bound and the achievable rate can be upper bounded as
\begin{align}
\Delta \leq 1+\frac{K+\log 3}{2K} <2.
\end{align}

\noindent \underline{\textbf{Case III.}} Next, we should examine the conditions in case III. In this case $R'_1=0$ and $R'_3$ is upper bounded by a constant. So we have
\begin{align}
R_{\mathrm{sym}}>R_2&= \cgf{\frac{\SNRP}{\INRP}+1}{2\INRP+1}\nonumber\\
&\geq \frac{1}{2}\log(1+\SNRP)-\frac{1}{2}\log 10.
\end{align}
Under this condition the upper bound in \eqref{eq:UB:G-VW} reduces to
\begin{align}
&\mathcal{C}_{\mathrm{sym},\mathrm{G}}^{\mathrm{FB}}(K)\nonumber\\
&\leq  (B-E) + \frac{E}{K}\nonumber\\
&\leq \frac{1}{2}\log\left(1+\INR+\frac{\SNRP+2\sqrt{\INR\cdot\SNRP}}{1+\INR}\right) + \frac{1}{6} \log 3\nonumber\\
&\leq \frac{1}{2}\log(1+\SNRP) +\frac{1}{2} \log 3 + + \frac{1}{6} \log 3,
\end{align}
where we used the facts that $K\geq 3$ and $\INR<2$ in the second inequality. Therefore,
\begin{align}
\Delta\leq \frac{1}{2}\log 10 + \frac{2}{3}\log 3 \leq 3.
\end{align}

\noindent \underline{\textbf{Case IV.}} Finally, in the last case $R'_1=R'_2=0$, and $R'_3$ is a constant, and we do not claim any positive rate  based on the proposed coding scheme. However, under this condition, the upper bound in \eqref{eq:UB:G-VW} would be
\begin{align}
\mathcal{C}_{\mathrm{sym},\mathrm{G}}^{\mathrm{FB}}(K)&\leq  (B-E) + \frac{E}{K}\nonumber\\
&\leq \frac{1}{2}\log \left(1+3\INRP^2+2\INRP+2\sqrt{2\INRP^3}\right)\nonumber\\
&\leq \frac{1}{2} \log 25 < \frac{5}{2},
\end{align}
and hence, the gap is bounded by $5/2$.

\subsection{Verification of Rate Bounds in Section~\ref{GaussianSection}}
\label{app: verify-bounds}
In this section we verify the constraints on the rates allocated to sub-messages in the coding scheme used for the Gaussian network. These constraints are due to decodability of the messages at different terminals based on the proposed 
decoding strategies. We verify the constraints for very weak interference regime, $0 \leq \alpha \leq 1/2$, and omit the details for other cases for sake of brevity. More precisely, will will show that the rates proposed in \eqref{eq:rates-vw} satisfy the inequalities in \eqref{const:VW-1}--\eqref{const:VW-4}, provided that $\INR\geq 2$ and $\SNR \geq 2\INR^2$.

It is trivial to see that $R_1$ in \eqref{eq:rates-vw} satisfies \eqref{const:VW-1}. Comparing $R_2$ in Ê\eqref{eq:rates-vw} to the right-hand side (RHS) of \eqref{const:VW-3} reveals that both expressions have identical denominators, while the nominator of \eqref{const:VW-3} has an extra additive $\INR$ term, which makes it larger than the proposed rate. Similarly, $R_3$ in Ê\eqref{eq:rates-vw} is always smaller than the RHS of \eqref{const:VW-4}.

It remains to verify \eqref{const:VW-2}. To this end, it suffices to show that 
\begin{align*}
\frac{\INR+1}{3} \leq \frac{\SNR+\INR+1}{\frac{\SNR}{\INR}+\INR+1},
\end{align*}
or equivalently, 
\begin{align*}
3(\SNR+\INR+1)- \left( \SNR+\INR^2+2\INR + \frac{\SNR}{\INR}+1\right)\geq 0,
\end{align*}
which can be further simplified to 
\begin{align*}
(\SNR-\INR^2) + \left(\SNR-\frac{\SNR}{\INR}\right)+ \INR+2\geq 0.
\end{align*}
Note that the latter is obvious due to the regime assumptions for the values of Ê$\SNR$ and $\INR$.
\section*{Acknowledgement}
We are grateful to the Associate Editor and the reviewers for their careful reading 
of the manuscript and helpful suggestions.
\bibliographystyle{IEEEtran}
\bibliography{refravi}
\begin{biographynophoto}{Ravi Tandon} (S03, M09) received the B.Tech degree
in electrical engineering from the Indian Institute
of Technology (IIT), Kanpur in 2004 and
the Ph.D. degree in electrical and computer engineering
from the University of Maryland, College
Park in 2010. From 2010 until 2012, he was
a post-doctoral research associate with Princeton
University. In 2012, he joined Virginia Polytechnic
Institute and State University (Virginia Tech)
at Blacksburg, where he is currently a Research
Assistant Professor in the Department of Electrical
and Computer Engineering. His research interests are in the areas of network information
theory, communication theory for wireless networks and information theoretic
security.

Dr. Tandon is a recipient of the Best Paper Award at the Communication
Theory symposium at the 2011 IEEE Global Communications Conference.
\end{biographynophoto}

\begin{biographynophoto}
{Soheil Mohajer} received the B.Sc. degree in electrical
engineering from the Sharif University of
Technology, Tehran, Iran, in 2004, and the M.Sc.
and Ph.D. degrees in communication systems both
from Ecole Polytechnique F\'{e}d\'{e}rale de Lausanne
(EPFL), Lausanne, Switzerland, in 2005 and 2010,
respectively. He then joined Princeton University,
New Jersey, as a post-doctoral research associate. Dr.
Mohajer has been a post-doctoral researcher at the
University of California at Berkeley, since October
2011.

His research interests include network information theory, data compression, wireless communication,  and bioinformatics.
\end{biographynophoto}

\begin{biographynophoto}
{H. Vincent Poor} (S72, M77, SM82, F87) received
the Ph.D. degree in electrical engineering and computer
science from Princeton University in 1977.
From 1977 until 1990, he was on the faculty of the
University of Illinois at Urbana-Champaign. Since
1990 he has been on the faculty at Princeton, where
he is the Dean of Engineering and Applied Science,
and the Michael Henry Strater University Professor
of Electrical Engineering. Dr. Poor's research interests
are in the areas of stochastic analysis, statistical
signal processing and information theory, and their
applications in wireless networks and related fields including social networks
and smart grid. Among his publications in these areas are 
{\it Smart Grid Communications and Networking} 
(Cambridge University Press, 2012) and {\it Principles of 
Cognitive Radio} (Cambridge University Press, 2013).

Dr. Poor is a member of the National Academy of Engineering and the
National Academy of Sciences, a Fellow of the American Academy of
Arts and Sciences, and an International Fellow of the Royal Academy of
Engineering (U. K.). He is also a Fellow of the Institute of Mathematical
Statistics, the Optical Society of America, and other organizations. In 1990,
he served as President of the IEEE Information Theory Society, in 2004-07
as the Editor-in-Chief of these {\it Transactions}, and in 2009 as General Co-chair
of the IEEE International Symposium on Information Theory, held in Seoul,
South Korea. He received a Guggenheim Fellowship in 2002 and the IEEE
Education Medal in 2005. Recent recognition of his work includes the 2010
IET Ambrose Fleming Medal for Achievement in Communications, the 2011
IEEE Eric E. Sumner Award, the 2011 IEEE Information Theory Paper Award,
and honorary doctorates from Aalborg University, the Hong Kong University
of Science and Technology, and the University of Edinburgh.
\end{biographynophoto}

\end{document}